\documentclass[a4paper,12pt]{article}
\usepackage[top=2.5cm,bottom=2.5cm,left=2.5cm,right=2.5cm]{geometry}
\usepackage[margin=1cm,%
font=small,%
format=hang,%
labelsep=period,%
labelfont=bf]{caption}
\usepackage[latin1]{inputenc}
\usepackage{amsmath}
\usepackage{amsfonts}
\usepackage{amssymb}
\usepackage{cite}
\usepackage{amsthm}
\usepackage{makeidx}
\usepackage{graphicx}
\usepackage{mathrsfs,xcolor}
\usepackage{enumerate}
\usepackage{blkarray}
\usepackage{bm}
\usepackage{setspace}
\usepackage{algorithm,algorithmic,refcount}
\usepackage{tikz-cd} 
\usepackage{lipsum}

\usepackage{appendix}

\usepackage{mathtools}
\newtagform{nonums}[\phantom]{}{} 

\makeatletter
\def\longleftrightarrowfill@{\arrowfill@\leftarrow\relbar\rightarrow}
\makeatother

\usepackage{tikz}
\usetikzlibrary{arrows,shapes,quotes}
\usetikzlibrary{patterns,decorations.pathreplacing}

\numberwithin{table}{section}
\numberwithin{equation}{section}

\theoremstyle{plain}
\newtheorem{theorem}{Theorem}
\newtheorem{lemma}{Lemma}
\newtheorem{proposition}{Proposition}
\newtheorem{corollary}{Corollary}

\theoremstyle{definition}
\newtheorem{definition}{Definition}
\newtheorem{example}{Example}

\usepackage{authblk}

\author[1,2]{\textbf{Lauro L. Fontanil}}
\author[2,3,4,5]{\textbf{Eduardo R. Mendoza}}

\affil[1]{\small \textit{Institute of Mathematical Sciences and Physics, University of the Philippines,  Los Ba\~{n}os, Laguna 4031, Philippines}}
\affil[2]{ \textit{Mathematics and Statistics Department, De La Salle University, Manila  0922, Philippines}}
\affil[3]{\normalsize\textit{Center for Natural Sciences and Environmental Research, De La Salle University, Manila  0922, Philippines}}
\affil[4]{\normalsize\textit{Max Planck Institute of Biochemistry, Martinsried near Munich, Germany}}
\affil[5]{\normalsize\textit{Faculty of Physics, Ludwig Maximilian University, Munich 80539, Germany}}
\affil[*]{Corresponding author: \texttt{lauro$\_$fontanil@dlsu.edu.ph}}

\title{\vspace{3.5cm}\textbf{Common Complexes of Decompositions and Complex Balanced Equilibria of \\Chemical Reaction Networks}}

\date{}

\begin{document}
\maketitle
\thispagestyle{empty}
\begin{abstract}

A decomposition of a chemical reaction network (CRN) is produced by partitioning its set of reactions. The partition induces networks, called subnetworks, that are ``smaller" than the given CRN which, at this point, can be called parent network. A complex is called a common complex if it occurs in at least two subnetworks in a decomposition. A decomposition is said to be incidence independent if the image of the incidence map of the parent network is the direct sum of the images of the subnetworks' incidence maps. It has been recently discovered that the complex balanced equilibria of the parent network and its subnetworks are fundamentally connected in an incidence independent decomposition. In this paper, we utilized the set of common complexes and a developed criterion to investigate decomposition's incidence independence properties. A framework was also developed to analyze decomposition classes with similar structure and incidence independence properties. We identified decomposition classes that can be characterized by their sets of common complexes and studied their incidence independence. Some of these decomposition classes occur in some biological and chemical models. Finally, a sufficient condition was obtained for the  complex  balancing  of some power law kinetic (PLK) systems with incidence independent and complex balanced decompositions. This condition led to a generalization of the Deficiency Zero Theorem for some PLK systems.

\end{abstract}


\section{Introduction}

Decompositions of chemical reaction networks (CRNs) were first studied by Feinberg in his 1987 review \cite{fein1987}. A decomposition is generated by a partition of the set of reactions, with the species and complexes of a subnetwork being those occurring in the corresponding reaction subset.  He also introduced the concept of an independent decomposition where the stoichiometric subspace of the parent network is the direct sum of the stoichiometric subspaces of the decomposition's subnetworks. He showed its importance by stating the following relationship among the positive equilibria of the network and those of the decomposition's subnetworks:

\begin{theorem}[Remark 5.4, \cite{fein1987}]\label{feinberg decomp} Let $(\mathscr N, K)$ be a chemical kinetic system with partition $\{\mathscr R_1, \cdots, \mathscr R_k\}$. If $\mathscr N=\mathscr N_1\cup \cdots \cup \mathscr N_k$ is the network decomposition generated by the partition and $E_+(\mathscr N_i, K_i)=\{x\in \mathbb R^{\mathscr S}_{>0}\;|\;N_iK_i(x)=0\}$, then 
\begin{enumerate} [i.]
    \item $\displaystyle \bigcap_{i=1}^k E_+(\mathscr N_i, K_i)\subseteq E_+(\mathscr N, K)$ and
    \item the equality holds if the network decomposition is independent.
\end{enumerate}
\end{theorem}

Feinberg's early result has turned out to be particularly useful in the new field of concentration robustness of kinetic systems, which he and Shinar initiated in a well-known paper in Science in 2010 \cite{shin2010}. A sufficient condition for deficiency one mass action systems they derived has, in the meantime, been extended to power law and Hill type kinetic systems $($\!\!\!\cite{fort2018}, \cite{hern2020}$)$. Independent decompositions allow the detection of concentration robustness in larger and higher deficiency kinetic system via Feinberg's theorem $($\!\!\!\cite{font2021}, \cite{fortal2021}, \cite{fort2021}, \cite{jose2021}$)$.

In view of the focus of the chemical reaction network theory (CRNT) literature on mass action systems, where the subset of complex balanced equilibria is either empty or coincides with the whole set of positive equilibria, the corresponding concept and result were formulated only recently by Farinas et al. \cite{fari2021}: a decomposition is incidence independent if the image of the incidence map of the parent network is the direct sum of the images of the incidence maps of the subnetworks. The relationship among the complex balanced equilibria sets is fully analogous to Feinberg's result:

\begin{theorem}[Theorem 4, \cite{fari2021}]\label{farinas decomp}
Let $\mathscr N=\mathscr N_1\cup \cdots \cup \mathscr N_k$ be a decomposition. Let $K$ be any kinetics, and $Z_+(\mathscr N, K)$ and $Z_+(\mathscr N_i, K_i)$ be the set of complex balanced equilibria of $\mathscr N$ and $\mathscr N_i$, respectively. Then, 
\begin{equation*}
     \displaystyle \bigcap_{i=1}^k Z_+(\mathscr N_i, K_i)\subseteq Z_+(\mathscr N, K).
\end{equation*}
If the decomposition is incidence independent, then
\begin{enumerate} [i.]
    \item $\displaystyle \bigcap_{i=1}^k Z_+(\mathscr N_i, K_i)=Z_+(\mathscr N, K)$ and
    \item $Z_+(\mathscr N, K)\neq \varnothing$ implies that $Z_+(\mathscr N_i, K_i)\neq \varnothing$ for $i=1, \cdots, k$. 
\end{enumerate}
\end{theorem}

The best-known example of an incidence independent decomposition is the set of linkage classes of a network. It also has the special property that it partitions not only the set of reactions but also the set of complexes of a network. Farinas et al. introduced the term $\mathscr C$-decomposition for a decomposition with the latter property and show that any $\mathscr C$-decomposition of a network is generated by a coarsening of the partition into reactions of linkage classes. Inspired by the study of a model of the gene regulatory system of {\em Mycobacterium tuberculosis} \cite{fari2020}, Farinas et al. also introduced the concept of a $\mathscr C^*$-decomposition, where the subnetworks pairwise have at most the zero complex. They also derived that any $\mathscr C^*$-decomposition is incidence independent.

These results motivated the systematic study of the set of common complexes of a decomposition, denoted by $\mathscr C_{\mathscr D}$ (the subscript ${\mathscr D}$ is the symbol assigned to the decomposition), and its relationships to incidence independence and complex balanced equilibria, which is the main topic of this paper. A common complex is a complex that occurs in at least two subnetworks in the decomposition. Clearly, $\mathscr C$-decompositions are those with $\mathscr C_{\mathscr D}=\varnothing$ (or equivalently $|\mathscr C_{\mathscr D}|= 0$), and $\mathscr C^*$-decompositions belong to those with $|\mathscr C_{\mathscr D}|\leq 1$. For decompositions with $|\mathscr C_{\mathscr D}| \geq 2$, incidence independence may not hold, and a criterion based on counting common complexes is established. Furthermore, a framework for describing decomposition classes based on the structure of $\mathscr C_{\mathscr D}$ is developed for formulating conditions for this property to hold for decomposition classes. The framework is illustrated in detail by the class of pairwise min-max (PMM) decompositions, which form a superset of the $\mathscr C^*$-decompositions.

Although the theorems of Feinberg and Farinas et al. determine the structure of the sets of positive equilibria and complex balanced equilibria through independent and incidence independent decompositions of kinetic systems respectively, neither result guarantees that the intersections are non-empty.  Specifically for an incidence independent decomposition, given that complex balancing of the subnetworks is a necessary condition for that of the whole system, the question is: which additional properties are needed for it to be sufficient too?

In \cite{fari2021}, Farinas et al. also established the surprisingly general result that the conclusion holds for any $\mathscr C$-decomposition, i.e., the additional purely structural condition of $\mathscr C_{\mathscr D}=\varnothing$ and any kinetic system.  We show with a simple example that this is no longer true for $\mathscr C_{\mathscr D}\neq \varnothing$. This result is singular in its purely structural character: as soon as $\mathscr C_{\mathscr D}\neq \varnothing$, there are kinetic systems with an incidence independent and complex balanced decomposition (a decomposition where each subnetwork is complex balanced) that are not complex balanced. In fact, mass action systems on the same network and the same incidence independent and complex balanced decomposition, may or may not be complex balanced. Hence, to obtain results when $\mathscr C_{\mathscr D}\neq \varnothing$ for an  incidence independent and complex balanced decomposition, one needs to restrict the set of kinetics systems for which the conclusion of complex balancing would hold. Ideally, if the latter is true for a kinetic system, it should also hold for all kinetic systems in its semi-module (i.e., for all kinetic systems with the same interaction function, differing only in rate constants). 

Accordingly, our result, while imposing no conditions on the set of common complexes of the incidence independent decomposition, is valid only for a class of power law kinetic (PLK) systems. The class is defined by two properties:
\begin{itemize}
    \item the subnetworks of the decomposition have power law reactant-determined kinetics (PL-RDK); and
    \item the induced decomposition of kinetic complexes is independent. 
\end{itemize}
\noindent The additional conditions involve properties of the kinetics on the decomposition subnetworks and associated structures.  

This result further underscores the importance of incidence independence for the complex balancing of the whole network. The result is analogous to a generalization of the Deficiency Zero Theorem (DZT) of Fortun et al. \cite{fort2019} for the set of all positive equilibria presented by Hernandez et al. \cite{hern2019}.

The main results of the paper include:

\begin{enumerate}[i.]
    \item a common complexes approach to study decomposition's properties, in particular incidence independence through a developed criterion;
    \item a framework for analyzing decomposition classes with similar structure and incidence independence properties illustrated in detail by the class of PMM decompositions, which contain the linkage class decomposition and occur in biological and chemical models; and
    \item a new condition ensuring the complex balancing of a weakly reversible power law kinetic system with an incidence independent decomposition into complex balanced subnetworks.  
\end{enumerate}

The paper is organized as follows. Fundamental concepts on chemical reaction networks, chemical kinetic systems, and decomposition theory are given in section 2. The framework to classify network's decompositions according to equivalent incidence independence properties is presented in section 3. In sections 4 and 5, the class of pairwise binary-sized (PBS) decompositions and its subclass of PMM decompositions are discussed. The condition ensuring complex balancing of a PLK system with an incidence independent complex balanced decomposition is derived in section 6. Section 7 provides a summary and an outlook.

\section{Fundamentals of chemical reaction networks and kinetic systems}

In this section, some necessary concepts and results on chemical reaction networks and chemical kinetic systems are reviewed, the details of which can be found in \cite{arce}, \cite{fein2019}, and \cite{toth}.

\subsection{Structure of chemical reaction networks}

 Consider a non-empty finite set $\mathscr S$ of \textbf{species}. A linear combination of these species with nonnegative integer coefficients (or \textbf{stoichiometric coefficients}) is called a \textbf{complex}. A complex whose stoichiometric coefficients are all zero is called the  \textbf{zero complex} and is denoted by $0$. The set of all complexes is denoted by $\mathscr C$. Here, $m$ and $n$ are used to denote the number of species and complexes, respectively. 

The set of \textbf{reactions}, denoted by $\mathscr R$, is a non-empty subset of $\mathscr C \times \mathscr C$ satisfying the conditions:
\begin{enumerate}[i.]
	\item $(y,y)\notin \mathscr R$ for all $y\in \mathscr C$; and
	\item for each $y\in \mathscr C$, there is a $y'\in \mathscr C$ such that $(y,y')\in \mathscr R$ or $(y',y)\in \mathscr R$.
\end{enumerate}
In a reaction $(y,y')$, which can also be written as $y\rightarrow y'$, $y$ and $y'$ are called \textbf{reactant} and \textbf{product complexes}, respectively. The order of $\mathscr R$ is denoted here by $p$.

\begin{definition} A \textbf{chemical reaction network} (CRN) $\mathscr N$ is a triple $(\mathscr S, \mathscr C, \mathscr R)$ of non-empty finite sets where $\mathscr S$, $\mathscr C$, and $\mathscr R$ are the sets of species, complexes, and reactions, respectively.
\end{definition}
\noindent The terms ``chemical reaction network'', ``reaction network'', and ``network'' are used here interchangeably. 

The \textbf{reactant map} $\rho: \mathscr R\rightarrow \mathscr C$ and the \textbf{product map} $\pi: \mathscr R\rightarrow \mathscr C$ map a reaction to its reactant and product complexes, respectively. The number of the reactant complexes is denoted by $n_{\rho}$, i.e., $|\rho(\mathscr R)|=n_{\rho}$. A CRN is called \textbf{cycle terminal} if $n_{\rho}=n$ while it is called \textbf{nonbranching} if $n_{\rho}=p$.

A CRN can be viewed as a digraph with the sets of complexes and reactions serve as the digraph's vertex and arc sets, respectively. The digraph representation of a reaction network, called \textbf{reaction graph}, is loopless and does not have an isolated vertex. The \textbf{underlying graph} of a reaction network is the graph obtained after removing the direction in each arc of its reaction graph. \textbf{Paths} and \textbf{cycles} in a CRN are the corresponding paths and cycles in the underlying graph. \textbf{Directed paths} and \textbf{directed cycles} in a reaction network are the uni-directed paths and cycles in the reaction graph.

\begin{example}\label{running}

Consider a toy CRN with the following digraph representation.
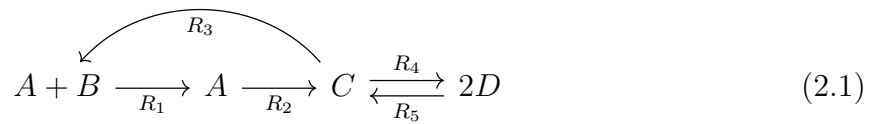
\begin{equation}
\begin{tikzcd}
A+B \arrow[r, "R_1"'] & A \arrow[r, "R_2"'] & C \arrow[r, "R_4", shift left] \arrow[ll, "R_3", bend right=49] & 2D \arrow[l, "R_5", shift left]
\end{tikzcd}
\label{graph100}
\end{equation}
\noindent This network has $m=4, n=4$, and $p=5$ such that
\begin{equation}
\mathscr{S} = \{ A,B,C,D\}, \mathscr{C}= \{A+B,A,C,2D \}, {\textnormal{\;and\;}} \mathscr{R}= \{R_1, R_2, R_3, R_4, R_5\}. 
\end{equation}

\noindent The network is also cycle terminal since $n=4=n_{\rho}$.

\end{example}

Two complexes are \textbf{connected} if there is a path joining the corresponding vertices in the underlying graph. A reaction network is said to be \textbf{weakly connected} if the underlying graph is connected, i.e., every pair of complexes are connected. Two complexes $y$ and $y'$ are said to be \textbf{strongly connected} if there is a directed path from $y$ to $y'$ and vice versa. The \textbf{linkage classes}, whose number is denoted by $\ell$, of a network corresponds to the maximal weakly connected subgraphs of its reaction graph while the \textbf{strong linkage classes} are the subgraphs of the reaction graph that are maximal strongly connected. If each linkage class is also a strong linkage class, a CRN is called \textbf{weakly reversible}. The toy CRN in Example \ref{running} is weakly reversible since its sole linkage class is a strong linkage class. 

CRNs are often studied with the aid of finite dimensional spaces $\mathbb R^{\mathscr S}$, $\mathbb R^{\mathscr C}$, and $\mathbb R^{\mathscr R}$, respectively referred to as \textbf{species space, complex space}, and \textbf{reaction space}. For every reaction $y\rightarrow y'$, a vector, called \textbf{reaction vector}, resulted from subtracting the reactant complex $y$ from the product complex $y'$ is associated. The linear subspace of $\mathbb R^{\mathscr S}$ defined by $S:=\text{span}\{y'-y\in \mathbb R^{\mathscr S}|y\rightarrow y'\in \mathscr R\}$ is called the \textbf{stoichiometric subspace} $S$ of a reaction network. Its dimension $s$ also refers to the \textbf{rank }of the network. A CRN can be characterized by a nonnegative integer $\delta$ called \textbf{deficiency} which is given by $\delta=n-\ell-s$. The toy CRN in Example \ref{running} has $S = \text{span}\{r_1, r_2, r_4\}$ where
\begin{equation}
    r_1 = \begin{bmatrix} 0\\-1\\0\\0\\ \end{bmatrix}, r_2 = \begin{bmatrix} -1\\0\\1\\0\\ \end{bmatrix}, \text{\;and\;} r_4 = \begin{bmatrix} 0\\0\\-1\\2\\ \end{bmatrix}.
\end{equation}
\noindent Hence, its rank $s=3$ while its deficiency is $\delta=4-1-3=0$, respectively.


\begin{definition} Let $\mathscr N = (\mathscr S, \mathscr C, \mathscr R)$ be a CRN. The \textbf{incidence map} $I_a: \mathbb R^{\mathscr R}\rightarrow \mathbb R^{\mathscr C}$ is the linear map defined by mapping for each reaction $R_i: y_i\rightarrow y'_i\in \mathscr R$, the basis vector $\omega_i$ to the vector $\omega_{y'_i}-\omega_{y_i}\in \mathscr C$.
\end{definition}

\noindent As a linear map, the incidence map has an $n\times p$ matrix representation, called \textbf{incidence matrix}, whose entries are described by the following:

\begin{center}
    $(I_a)_{(i,j)}= \left\{ \begin{array}{rl}
    -1 & \mbox{if $i$ is the reactant complex of reaction $j \in \mathscr R$} \\ 1 & \mbox{if $i$ is the product complex of reaction $j \in \mathscr R$} \\ 0 & \mbox{otherwise} \\
\end{array}\right.$
\end{center}

\noindent The toy CRN in Example \ref{running} has the following incidence matrix.

\begin{equation}
I_a = 
    \begin{blockarray}{*{5}{c} c}
        \begin{block}{*{5}{>{$\footnotesize}c<{$}} c}
            $R_1$ & $R_2$ & $R_3$ & $R_4$ & $R_5$ &  \\
        \end{block}
        \begin{block}{[*{5}{c}]>{$\footnotesize}c<{$}}
            -1& 0 & 1 & 0 & 0 &    $A+B$\\
            1 & -1& 0 & 0 & 0 &    $A$\\
            0 & 1 & -1& -1& 1 &    $C$\\
            0 & 0 & 0 & 1 & -1 &    $2D$\\
        \end{block}
    \end{blockarray}
\label{eqn9}
\end{equation}

\noindent Boros \cite{boros2013} described the rank of the incidence map through the following proposition.

\begin{proposition}[Proposition 4.2.3, \cite{boros2013}]\label{boros} $\dim \textnormal{Im} I_a=n-\ell$.
\end{proposition}

\subsection{Dynamics of chemical reaction networks}

A \textbf{kinetics} is an assignment of a rate function to each reaction in a CRN. A network $\mathscr N$ together with a kinetics $K$ is called a \textbf{chemical kinetic system} (CKS) and is denoted here by $(\mathscr N,K)$. \textbf{Power law kinetics} (PLK) is identified by the \textbf{kinetic order matrix}, which is a $p\times m$ matrix $F=[F_{ij}]$,
and vector $k\in \mathbb R^{\mathscr R}_{>0}$ called the \textbf{rate vector}.

\begin{definition} A kinetics $K: \mathbb R^{\mathscr S}_{>0} \rightarrow \mathbb R^{\mathscr R}$ is a \textbf{power law kinetics} if
\begin{equation}
   K_i(x)=k_ix^{F_{i,\cdot}} \text{\;for\;} i=1,\dots, r 
\end{equation}

\noindent where $k_i\in \mathbb R_{>0}$, $F_{i,j} \in \mathbb R$, and $F_{i,\cdot}$ is the row of $F$ associated to reaction $R_i$.
\end{definition}

A PLK system can be classified based on the kinetic orders assigned to its \textbf{branching reactions} (i.e., reactions sharing a common reactant complex). 

\begin{definition}
A PLK system has a \textbf{reactant-determined kinetics} (type PL-RDK) if for any two branching reactions $R_i, R_j\in \mathscr R$, the corresponding rows of kinetic orders in $F$ are identical, i.e., $F_{ih}=F_{jh}$ for $h=1, \dots,m$. Otherwise, a PLK system has a \textbf{non-reactant-determined kinetics} (type PL-NDK).
\end{definition}

The well-known \textbf{mass action kinetic} (MAK) system  forms a subset of PL-RDK systems. In particular, MAK is given by $K_i(x)=k_ix^{Y_{.,j}}$ for all reactions $R_i: y_i \rightarrow y'_i \in \mathscr R$ with $k_i\in \mathbb R_{>0}$ (called \textbf{rate constant}). The vector $Y_{.,j}$ contains the stoichiometric coefficients of a reactant complex $y_i\in \mathscr C$.

\begin{definition} The \textbf{species formation rate function} of a chemical kinetic system is the vector field
\begin{equation}
    f(c)=NK(c)=\displaystyle \sum_{y_i\rightarrow y'_i\in \mathscr R}K_i(c)(y'_i-y_i)
\end{equation}
where $c\in \mathbb R^{\mathscr S}_{\geq 0}$ and $N$ is the $m\times p$ matrix, called \textbf{stoichiometric matrix},  whose columns are the reaction vectors of the system. 
\end{definition}

\noindent The equation $dc/dt=f(c(t))$ is the \textbf{ODE or dynamical system} of the chemical kinetic system. An element $c^*$ of $\mathbb R^{\mathscr S}_{>0}$ such that $f(c^*)=0$ is called a \textbf{positive equilibrium} or \textbf{steady state} of the system. The set of all positive equilibria of a CKS is denoted by $E_+(\mathscr N,K)$.

Analogous to the species formation rate function, we also have the complex formation rate function.

\begin{definition} The \textbf{complex formation rate function} $g: \mathbb R^{\mathscr S}_{>0}\rightarrow \mathbb R^{\mathscr C}$ of a chemical kinetic system is the vector field
\begin{equation}
g(c)=I_aK(c)=\displaystyle \sum_{y_i\rightarrow y'_i\in \mathscr R}K_i(c)(\omega_{y'_i}-\omega_{y_i}),    
\end{equation}
\noindent where $c\in \mathbb R^{\mathscr S}_{\geq 0}$ and $I_a$ is the incidence map. 
\end{definition}
\noindent A CKS is \textbf{complex balanced} if it has complex balanced steady state, i.e., there is a composition $c^{**}\in \mathbb R_{>0}^{\mathscr S}$ such that $g(c^{**} )=0$. The set of all complex balanced steady states of the system is denoted by $Z_+(\mathscr N,K)$. 



\subsection{Decomposition theory}

Some basic concepts and results on decomposition theory are discussed in this subsection. More detailed discussions than the one contained here can be found in \cite{fari2021} and \cite{fort2021}.

We state the concept of restriction on the objects of a reaction network as given in \cite{josh2013} and \cite{josh2012}. Let $\mathscr N=(\mathscr S, \mathscr C, \mathscr R)$ be a given network and consider a subset $\mathscr C'\subseteq \mathscr C$ and a subset $\mathscr R'\subseteq \mathscr R$. The \textbf{restriction of} $\mathscr C$ \textbf{to} $\mathscr R'$, denoted by $\mathscr C|_{\mathscr R'}$, is the set of (reactant and product) complexes of the reactions in $\mathscr R'$. Likewise, the \textbf{restriction of} $\mathscr S$ \textbf{to} $\mathscr C'$, denoted by $\mathscr S|_{\mathscr C'}$, is the set of species that occur in the complexes in $\mathscr C'$.

Let $\mathscr N=(\mathscr S, \mathscr C, \mathscr R)$ be a given network. $\mathscr N'=(\mathscr S', \mathscr C', \mathscr R')$ is said to be a \textbf{subnetwork} of $\mathscr N$ if $\mathscr R'\subseteq \mathscr R$. In this notion, $\mathscr C'=\mathscr C|_{\mathscr R'}$ and $\mathscr S'=\mathscr S|_{\mathscr C'}$. That is, $\mathscr N'$ has complex set consists of the complexes of the reactions in $\mathscr R'$ and has species set consists of the species that occur in the complexes of $\mathscr C|_{\mathscr R'}$.

\begin{definition} Let $\mathscr N = (\mathscr S, \mathscr C, \mathscr R)$ be a CRN. A \textbf{covering} of $\mathscr N$ is a collection of subsets \{$\mathscr R_1, \dots, \mathscr R_k$\} whose union is $\mathscr R$. A covering is called a \textbf{decomposition} of $\mathscr N$ if $\mathscr R_i$'s form a partition of $\mathscr R$, i.e., $\mathscr R_i$'s are pairwise disjoint.
\end{definition}

As previously indicated, $\mathscr R_i$ defines a subnetwork $\mathscr N_i$ of $\mathscr N$ such that $\mathscr N_i=(\mathscr S_i, \mathscr C_i, \mathscr R_i)$ where $\mathscr C_i=\mathscr C|_{\mathscr R_i}$ and $\mathscr S_i=\mathscr S|_{\mathscr C_i}$. Hence, 
\begin{equation} \label{eq4}
	\mathscr R=\bigcup \mathscr R_i \Rightarrow \mathscr C=\bigcup \mathscr C_i \Rightarrow \mathscr S=\bigcup \mathscr S_i.   
\end{equation}

Gross et al. \cite{gros2020} defined the {union of networks} as follows. 

\begin{definition}\label{union}
	The \textbf{union of networks} $\mathscr N_i=(\mathscr S_i, \mathscr C_i, \mathscr R_i)$, $i=1,\cdots,k$, is given by
	
\begin{equation}
    \mathscr N_1\cup\cdots\cup \mathscr N_k:=\Big(\displaystyle \bigcup_{i=1}^k \mathscr S_i,\displaystyle \bigcup_{i=1}^k \mathscr C_i, \displaystyle \bigcup_{i=1}^k \mathscr R_i\Big)
\end{equation}
\end{definition}

\noindent Because of the clear connection of Definition \ref{union} and (\ref{eq4}), we adopt the notation for union of networks to indicate covering or decomposition of a network and write $\mathscr N=\mathscr N_1\cup \cdots \cup\mathscr N_k$ to mean that $\{\mathscr R_1, \dots, \mathscr R_k\}$ is a covering or a decomposition of $\mathscr N$. 

Given a decomposition $\mathscr N=\mathscr N_1\cup \cdots \cup\mathscr N_k$, we may refer to $\mathscr N$ as the \textbf{parent network} of its subnetworks $\mathscr N_i$'s. The number of subnetworks in a decomposition (usually denoted by $k$) is called the \textbf{length of the decomposition}.  If the length of a decomposition is 0, it is called an \textbf{empty decomposition}. A decomposition is called \textbf{trivial} if its length is $1$, i.e., it consists of only one subnetwork. Here, a decomposition is assumed to be non-empty and non-trivial unless otherwise stated. We extend the CRN notations used in this paper to fit the concept of subnetwork. The complete list of notations adapted for the subnetworks used in this paper is given in the appendix.

\textbf{Independent subnetworks}, as given in \cite{fein1987}, are obtained from two complementary subsets of the reaction set which satisfy the condition that the stoichiometric subspace of the parent network equals the direct sum of the corresponding stoichiometric subspaces of the subnetworks. The notion of independent subnetworks was the early form of independent decomposition defined as follows.

\begin{definition}
	A decomposition is \textbf{independent} if $S$ is the direct sum of the subnetworks' stoichiometric subspaces $S_i$ or equivalently, $s=s_1+\cdots+s_k$.
\end{definition}

\begin{example}\label{ex14} Let the toy CRN in Example \ref{running} be denoted by $\mathscr N$. Consider the subnetworks $\mathscr N_1$ and $\mathscr N_2$ of $\mathscr N$ whose reaction graphs are given below. 
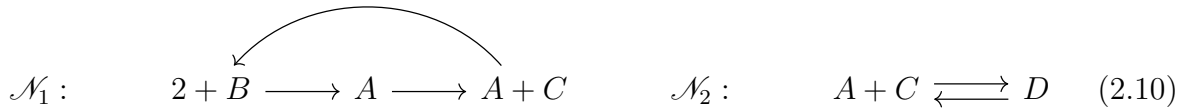
\vspace{-0.1cm}	
	\begin{equation}\label{eq55}
\begin{tikzcd}
\mathscr N_1: & 2+B \arrow[r] & A \arrow[r] & A+C \arrow[ll, bend right=49] & \mathscr N_2: & A+C \arrow[r, shift left] & D \arrow[l, 
shift left]
\end{tikzcd}
	\end{equation}
\noindent It is clear that $\mathscr N=\mathscr N_1\cup \mathscr N_2$ is a decomposition. Moreover, $\mathscr N_1$ and $\mathscr N_2$ have ranks $2$ and $1$, respectively. Since the parent network is of rank $3$, the decomposition is independent.
\end{example}

Farinas et al. \cite{fari2021} introduced the concept of {incidence independent decomposition} that is patterned after the independent decomposition but utilizes the images of the incidence maps instead of the stoichiometric subspaces.

\begin{definition} A decomposition $\mathscr N=\mathscr N_1\cup \cdots \cup\mathscr N_k$  
	is \textbf{incidence independent} if and only if the image of the incidence map of $\mathscr N$ is the direct sum of the images of the incidence maps of $\mathscr N_i$'s.
\end{definition}

\noindent Due to Proposition \ref{boros}, a decomposition is incidence independent if and only if $n-\ell = \sum (n_i-\ell_i)$ where $n_i$ and $\ell_i$ are the number of complexes and linkage classes of $\mathscr N_i$. 

\begin{example} \label{ex15}
	The decomposition given in Example \ref{ex14} is also incidence independent since $n-\ell=4-1=3$ while $n_1-\ell_1=3-1=2$ and $n_2-\ell_2=2-1=1$.
\end{example}


$\mathscr C$- and $\mathscr C^*$-decompositions are decomposition classes defined in \cite{fari2021} which are found to be useful in several recent studies (\!\!\cite{fort2021},\cite{fari2020}). 

\begin{definition}
	A decomposition $\mathscr N=\mathscr N_1\cup \cdots \cup\mathscr N_k$ with $\mathscr N_i=(\mathscr S_i,\mathscr C_i,\mathscr R_i)$ is a \textbf{$\mathscr C$}-\textbf{decomposition} if $\mathscr C_i\cap \mathscr C_j=\varnothing$ for $i\neq j$.
\end{definition}

$\mathscr C$-decomposition is a decomposition where the set of complexes of the parent network is partitioned to become the complex sets of the subnetworks. The following structure theorem identifies $\mathscr C$-decomposition as a coarsening of the linkage class decomposition.

\begin{theorem}[Structure Theorem for $\mathscr C$-decomposition, \cite{fari2021}] \label{str_thm_for_C} $\mathscr N=\mathscr N_1\cup \cdots \cup \mathscr N_k$ is a $\mathscr C$-decomposition if and only if each $\mathscr N_i$ is the union of linkage classes and each linkage class is contained in only one $\mathscr N_i$.
\end{theorem}

\noindent It immediately follows from this result that, in any $\mathscr C$-decomposition, $\sum \ell_i=\ell$. With this result, the incidence independence of a $\mathscr C$-decomposition is guaranteed as stated in the following proposition. 

\begin{proposition}[Proposition 3, \cite{fari2021}] \label{C_dec_inc_ind}
	Any $\mathscr C$-decomposition is incidence independent.
\end{proposition}

On the other hand, $\mathscr C^*$-decomposition is characterized by the condition that the pairwise intersection of the subnetworks' complex sets can only contain the zero complex.

\begin{definition}
	A decomposition $\mathscr N=\mathscr N_1\cup \cdots \cup\mathscr N_k$ with $\mathscr N_i=(\mathscr S_i,\mathscr C_i,\mathscr R_i)$ is a \textbf{$\mathscr C^*$}-\textbf{decomposition} if $\mathscr C^*_i\cap \mathscr C^*_j=\varnothing$ for $i\neq j$ where $\mathscr C^*_i$ and $\mathscr C^*_j$ are the sets of  non-zero complexes in $\mathscr C_i$ and $\mathscr C_j$, respectively.
\end{definition}

Equivalently, in $\mathscr C^*$-decomposition, $\mathscr C_i \cap \mathscr C_j\subseteq \{0\}$. From this, it follows that $\mathscr C$-decompositions are contained in the class of $\mathscr C^*$-decompositions. An analogous structure theorem for $\mathscr C^*$-decomposition was also provided in \cite{fari2021}. 

\begin{theorem}[Structure Theorem for $\mathscr C^*$-decomposition, \cite{fari2021}] \label{str_thm_for_C*}  
	Let $\mathscr N=\mathscr N_1\cup \cdots \cup \mathscr N_k$ be a $\mathscr C^*$-decomposition and $\mathscr L_{0}$ and $\mathscr L_{0,i}$ be the linkage classes of $\mathscr N$ and $\mathscr N_i$ containing the zero complex $($note that $\mathscr L_{0,i}$ is empty if $\mathscr N_i$ does not contain the zero complex$)$. Then,
	\begin{enumerate}[i.]
		\item the $\mathscr L_{0,i}$ form a $\mathscr C^*$-decomposition of $\mathscr L_0$ and
		\item the $($non-empty$)$ $\mathscr N_i\setminus \mathscr L_{0,i}$ form a $\mathscr C$-decomposition of $\mathscr N\setminus \mathscr L_0$.
	\end{enumerate}
\end{theorem}

A criterion was formulated in \cite{fari2021} for the incidence independence of a $\mathscr C^*$-decomposition that is not $\mathscr C$-decomposition: 
\begin{equation}\label{com_com_cri}
    \sum\ell_i -\ell= k(1)-1
\end{equation}
\noindent where $k(1)$ is the number of subnetworks $\mathscr N_i$ that contain the zero complex. This criterion play an important role in proving the incidence independence of $\mathscr C^*$-decompositions.

\begin{proposition}[Corollary 5, \cite{fari2021}]
Any $\mathscr C^*$-decomposition is incidence independent.
\end{proposition}

\section{The common complexes framework for decomposition classes and their incidence independence properties}

In this section we introduce the concept of the set of common complexes of a decomposition, denoted by $\mathscr C_{\mathscr D}$, and the associated properties which will be needed in the succeeding sections in investigating the incidence independence of decompositions.

\subsection{Definition of $\mathscr C_{\mathscr D}$ and some basic properties} 

We define the set of common complexes as follows.

\begin{definition}
Let $\mathscr D$ be a decomposition. The \textbf{set} $\mathscr C_{\mathscr D}$  \textbf{of common complexes} consists of all complexes contained in sets of complexes of at least two distinct subnetworks of $\mathscr D$. 
\end{definition}

\noindent The order of $\mathscr C_{\mathscr D}$ is denoted here by $d$, i.e., $|\mathscr C_{\mathscr D}|=d$. The following proposition indicates that specifying pairwise intersections of the subnetworks' complex sets determines the set of common complexes.

\begin{proposition}\label{cott}
	In decomposition $\mathscr D$, $\mathscr C_{\mathscr D}={\displaystyle\bigcup_{i<j}}{(\mathscr C_i\cap \mathscr C_j)}$.
\end{proposition}

Note that for $i\neq j$, $(\mathscr C_i\setminus \mathscr C_{\mathscr D})\cap (\mathscr C_j\setminus \mathscr C_{\mathscr D})=\varnothing$. Moreover, in $\mathscr C$-decomposition, $\mathscr C_{\mathscr D}=\varnothing$ while $\mathscr C_{\mathscr D}=\{0\}$ for $\mathscr C^*$-decomposition that is not a $\mathscr C$-decomposition. Since every $\mathscr C$-decomposition is incidence independent, all decompositions with $d=0$ is incidence independent. It is indicated in Proposition \ref{prop7} that a decomposition with $d=1$ is incidence independent. Here is a counterexample showing that incidence independence does not necessarily follow, in general, for $d \geq 2$.

\begin{example}\label{sim_exa}
	Consider the decomposition $\mathscr D:\mathscr N=\mathscr N_1 \cup \mathscr N_2$ given by the following
	 \begin{equation}
		\begin{tikzcd}
			\mathscr N: A\leftrightarrow B & \mathscr N_1: A \rightarrow B & \mathscr N_2: A\leftarrow B
		\end{tikzcd}
		\label{cecil32}
	\end{equation}

\noindent Notice that $\mathscr C_{\mathscr D}=\{A,B\}$ and $n-\ell=2-1=1$ while $n_1-\ell_1=n_2-\ell_2=2-1=1$. Hence, $\mathscr D$ is not incidence independent.
\end{example}

Let $h_i$ be the number of complexes in subnetwork $\mathscr N_i$ that are also contained in $\mathscr C_{\mathscr D}$, i.e., $h_i=|\mathscr C_i\cap \mathscr C_{\mathscr D}|$. The following result, which generalizes (\ref{com_com_cri}), provides a criterion that determines the incidence independence of a decomposition through its $\mathscr C_{\mathscr D}$. We call this \textbf{connectivity-common complexes criterion} ({3C-criterion}). 

\begin{proposition}\label{3C_cri_PBS}
	Let $\mathscr D:\mathscr N = \mathscr N_1\cup \cdots\cup \mathscr N_k$ be a decomposition with $d=|\mathscr C_{\mathscr D}|$. Then, the decomposition is incidence independent if and only if $\displaystyle \sum^k_{i=1}\ell_i - \ell = \sum^k_{i=1} h_i-d$.
\end{proposition} 

\begin{proof}
Observe that $n = \displaystyle d+\sum^{k}_{i=1} (n_i-h_i)$. Then, we have
\begin{center}
\begin{tabular}{ ll } 
incidence independent decomposition  & $\Leftrightarrow n-\ell = \displaystyle \sum^k_{i=1}(n_i-\ell_i) = \sum^k_{i=1} n_i -  \sum^k_{i=1} \ell_i$ \\ 
 & $\displaystyle \Leftrightarrow n-\ell =  (d-d)+\sum^k_{i=1} (n_i-h_i+h_i) -  \sum^k_{i=1} \ell_i$\\
 & $\displaystyle \Leftrightarrow n-\ell =  (d-d)+\sum^k_{i=1} (n_i-h_i)+\sum^k_{i=1}h_i -  \sum^k_{i=1} \ell_i$\\
 & $\Leftrightarrow \displaystyle \sum^k_{i=1}\ell_i - \ell = \sum^k_{i=1} h_i-d$\\
\end{tabular}    
\end{center}
\end{proof}

\begin{example}
The criterion for the incidence independence of a $\mathscr C^*$-decomposition that is not $\mathscr C$-decomposition (\ref{com_com_cri}) is obtained from the 3C-criterion when $d=1$: $\displaystyle \sum^k_{i=1}\ell_i - \ell = \sum^{k(1)}_{i=1}1 - 1=k(1)-1$. 
\end{example}



\subsection{The common complexes framework of a decomposition class}

We now introduce a framework called \textbf{Common Complexes Framework} ({CCF}) that will provide a way to classify network's decompositions into sets with equivalent incidence independence properties.

For any positive integer $k$, an incidence independence class of decompositions of $\mathscr N$ with $k$ subnetworks is given by a set of conditions (called \textbf{class conditions}) defining the $k(k- 1)/2$ pairwise intersections $\mathscr C_i \cap \mathscr C_j$, $i < j$, for each decomposition. The class conditions determine the range of (integer) values for $0 \leq d < |\mathscr C|$. The $k \times k$ upper triangular matrix (with zero diagonal) $\mathscr M$ given by 
\begin{equation}
    m_{ij} = \begin{cases} 0 & \mbox{if } i\geq j \\ |\mathscr C_i \cap \mathscr C_j| & \mbox{if } i<j \end{cases}
\end{equation}
\noindent is called the \textbf{pairwise intersection matrix} (PIM). 

CCF has the following three components:

\begin{enumerate}[i.]
	\item a class condition that determines the PIM and the relevant parameters;
	\item a structure theorem for the decompositions in the class in terms of $\mathscr C$-decompositions; and
	\item a formulation of the 3C-criterion for incidence independence in terms of the class parameters.
\end{enumerate}
\noindent PIM were chosen to be a matrix with elements of upper triangular nonnegative matrices which form an additive monoid and allows scalar multiplication with positive integers. These operations may have interpretable effects on the decompositions involved. The structure theorem and 3C-criterion are often useful in finding network properties that are sufficient or necessary for the incidence independence of a decomposition.

\begin{example}
	The class of $\mathscr C$-decompositions can be defined through the following equivalent statements
	\begin{equation}
		\mathscr C_i\cap \mathscr C_j=\varnothing\text{\;for\;} i<j\Leftrightarrow d=0 \Leftrightarrow \mathscr M = \mathbf 0
	\end{equation}
	\noindent where $\mathbf 0$ is the zero matrix. The structure theorem (Theorem \ref{str_thm_for_C}) states that every $\mathscr C$-decomposition is a coarsening of the linkage class decomposition which led to the 3C-criterion given by $\sum\ell_i-\ell=0$. 
\end{example}

\begin{example}
	For $\mathscr C^*$-decomposition that is not a $\mathscr C$-decomposition, $d=1$ and the PIM is the matrix with entries $0$ and $1$ (for $ij$-entry corresponding to $\mathscr C_i$ and $\mathscr C_j$ whose intersection is the zero complex). See Theorem \ref{str_thm_for_C*} and (\ref{com_com_cri}) for the structure theorem and 3C-criterion.
\end{example}

\section{The class of pairwise binary-sized decompositions}

In this section, we discuss a class of network decompositions characterized by the property that the intersection sets of subnetworks' complex sets take on two sizes which are denoted by $b$ and $c$.

\subsection{Definition and examples}

\begin{definition}
Let $\mathscr D: \mathscr N=\mathscr N_1\cup \cdots \cup \mathscr N_k$ be a decomposition and $b$ and $c$ be integers such that $0\leq b\leq c\leq d=|\mathscr C_{\mathscr D}|$. Then, $\mathscr D$ is said to be \textbf{pairwise binary-sized} (PBS) if $|\mathscr C_i\cap \mathscr C_j|=b$ or $c$ for $i\neq j$. 
\end{definition}

\noindent PBS decompositions for a given pair $(b,c)$ are denoted here by $\text{PBS}_d(b,c)$. In the context of this paper, $b$ and $c$ are often zero and non-zero, respectively. Hence, the number $b$ can be thought of as the ``bottom number" of pairwise intersection elements while $c$ is the common cardinality of the pairwise intersection set. 

\begin{example} Observe that $\mathscr C^*$-decompositions that are not $\mathscr C$-decompositions are $\text{PBS}_1(0,1)$ or $\text{PBS}_1(1,1)$  while $\mathscr C$-decompositions are $\text{PBS}_0(0,0)$. Note that $d=0 \Leftrightarrow b=0 \textnormal{\;and\;} c=0$. Hence, $\mathscr C$-decompositions can be denoted simply with $\text{PBS}_0$.
\end{example}

We now provide several examples of PBS decompositions in models of chemical and biological systems.

\begin{example}\label{carbon cycle model}
Take the following weakly reversible and deficiency zero subnetwork $\mathscr N$ of the Schmitz's pre-industrial carbon cycle model \cite{schm2002}. 
 
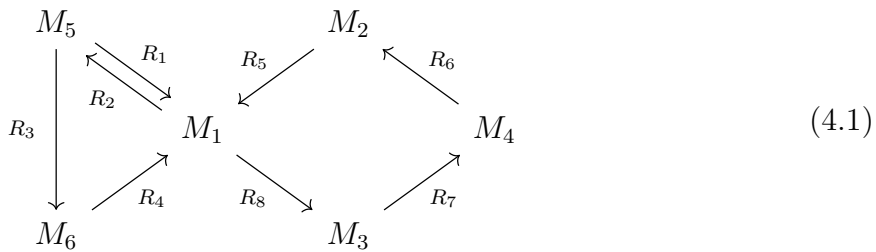
\begin{equation}
\begin{tikzcd}
M_5 \arrow[dd, "R_3\;\;"'] \arrow[rd, "R_1", shift left] &                                                      & M_2 \arrow[ld, "R_5"'] &                        \\
                                                         & M_1 \arrow[lu, "R_2", shift left] \arrow[rd, "R_8"'] &                        & M_4 \arrow[lu, "R_6"'] \\
M_6 \arrow[ru, "R_4"']                                   &                                                      & M_3 \arrow[ru, "R_7"'] &                       
\end{tikzcd}
\label{graph1000}
\end{equation}

\noindent The complexes in this network represent six carbon pools (see the table below) while the reactions describe the mass transfers. Fortun et al. \cite{fort2018} studied a power law representation of this model in \cite{fort2018}.

\begin{table}[H]
\centering
\begin{tabular}{|l|c|ll}
\cline{1-2}
$M_1$ & atmosphere               &  &  \\ \cline{1-2}
$M_2$ & warm ocean surface water &  &  \\ \cline{1-2}
$M_3$ & cool ocean surface water &  &  \\ \cline{1-2}
$M_4$ & deep ocean waters        &  &  \\ \cline{1-2}
$M_5$ & terrestrial biota        &  &  \\ \cline{1-2}
$M_6$ & soil and detritus        &  &  \\ \cline{1-2}
\end{tabular}
\caption{The carbon pools of the Schmitz's pre-industrial carbon cycle model.}
\end{table}    

\noindent  Now, consider the decomposition $\mathscr D_1: \mathscr N=\mathscr N_1\cup\mathscr N_2$ (also given in \cite{fort2018}) where

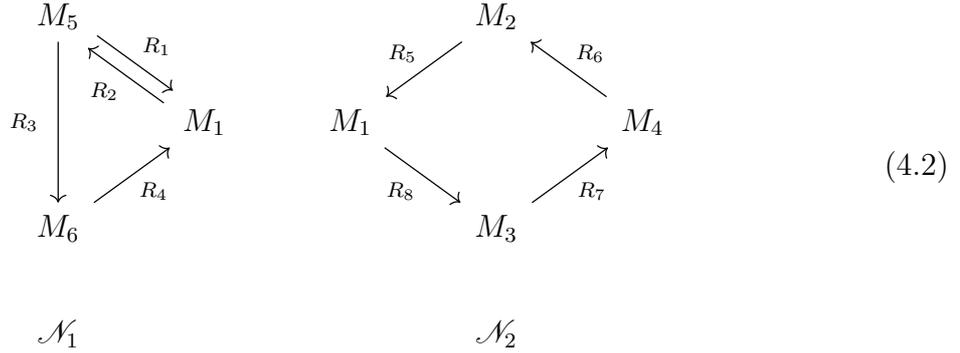
	\begin{equation}\label{eq5}
\begin{tikzcd}
M_5 \arrow[dd, "R_3\;\;"'] \arrow[rd, "R_1", shift left] &                                   &                        & M_2 \arrow[ld, "R_5"'] &                        \\
                                                         & M_1 \arrow[lu, "R_2", shift left] & M_1 \arrow[rd, "R_8"'] &                        & M_4 \arrow[lu, "R_6"'] \\
M_6 \arrow[ru, "R_4"']                                   &                                   &                        & M_3 \arrow[ru, "R_7"'] &                        \\
\mathscr N_1                                             &                                   &                        & \mathscr N_2           &                       
\end{tikzcd}
	\end{equation}

\noindent Note that $\mathscr D_1$ is incidence independent according to Proposition \ref{3C_cri_PBS} since $(\ell_1+\ell_2) - \ell=(1+1)-1=1$ and $(h_1+h_2)-d=(1+1)-1=1$. In addition, $\mathscr D_1$ is a PBS decomposition with $\mathscr C_{\mathscr D_1}=\{M_1\}$ and belongs to $\text{PBS}_1(1,1)$.
\end{example}

\begin{example}
The $\mathscr C^*$-decomposition of a reaction network modeling the gene regulatory system of \textit{Mycobacterium tuberculosis} in \cite{fari2020} is a $\text{PBS}_1(1,1)$ decomposition.     
\end{example}

\begin{example}\label{q-site}
The reaction graph below represents the CRN $\mathscr N$ of $q$-site distributive phosphorylation/dephosphorylation ($q\geq 2$)\cite{conr2018}.
 
\begin{equation*}\label{pd}
    \begin{tikzcd}
S_0+K \arrow[r, shift left] & S_0K \arrow[l, shift left] \arrow[r]   & S_1+K \arrow[r, shift left] & S_1K \arrow[l, shift left] \arrow[r, shift left] & S_2+K \arrow[l, shift left] \arrow[r, shift left] & \cdots \arrow[r] \arrow[l, shift left] & S_q+K \\
Sq+F \arrow[r, shift left]  & \cdots \arrow[l, shift left] \arrow[r] & S_2+F \arrow[r, shift left] & S_2F \arrow[l, shift left] \arrow[r]             & S_1+F \arrow[r, shift left]                       & S_1F \arrow[l, shift left] \arrow[r]   & S_0+F
\end{tikzcd}
\end{equation*}
\end{example}
\vspace{0.1cm}

\noindent The $S_i$'s are the $q$ substrates of a protein (the sites for both processes) while $K$ and $F$ are two different enzymes effecting phosphorylation and dephosphorylation respectively. Hernandez et al.\cite{hern2020A} gave the decomposition $\mathscr D_2:\mathscr N=\mathscr N_1\cup \cdots \cup \mathscr N_q$ where

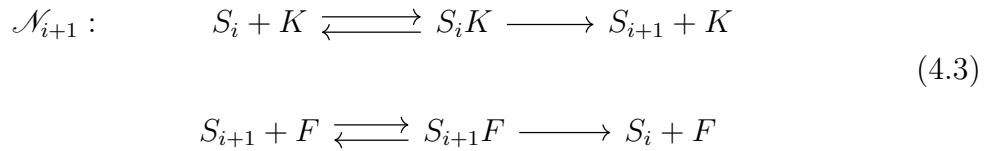
\begin{equation}\label{fundamental}
   \begin{tikzcd} 
   \mathscr N_{i+1}: & S_i+K \arrow[r, shift left]     & S_iK \arrow[l, shift left] \arrow[r]     & S_{i+1}+K \\
                  & S_{i+1}+F \arrow[r, shift left] & S_{i+1}F \arrow[l, shift left] \arrow[r] & S_i+F    
\end{tikzcd}
\end{equation}
\vspace{0.1cm}

\noindent for $i=0,1,\cdots,q-1$. They call $\mathscr D_2$ a fundamental decomposition (in the sense of Ji \cite{ji2011}). $\mathscr D_2$ is a  $\text{PBS}_d(0,2)$ decomposition with $\mathscr C_{\mathscr D_2}=\{S_{i+1}+K, S_{i+1}+F\;|\;i=0,1,\cdots, q-2\}$ and $d=2(q-1)$. Note that $\mathscr D_2$ is incidence independent (Proposition 3.38 in \cite{hern2020A}). We verify this using Proposition \ref{3C_cri_PBS}. Notice that  $\displaystyle \sum^q_{i=1}\ell_i - \ell = 2q-2$ and $\displaystyle \sum^{q}_{i=1}h_i-d=\Bigg{[}h_1+\sum^{q-1}_{i=2}h_i+h_q\Bigg{]}-d=[2+4(q-2)+2]-[2(q-1)]=2q-2$.

In a PBS decomposition, we denote by $k(c)$ the number of subnetworks containing the ``non-bottom" common pairwise intersection. As a simple illustration, a $\mathscr C$-decomposition has  $k(c)=0$. For convenience, unless otherwise indicated, let the first $k(c)$ subnetworks of a PBS decomposition (i.e., $\mathscr N_1, \cdots, \mathscr N_{k(c)}$) be the ones that contain the non-bottom common pairwise intersection. With this, the corresponding PIM $\mathscr M$ is the matrix $[m_{ij}]$ where
\begin{equation}
    m_{ij} = \begin{cases} c & \mbox{if } i< j \leq k(c)\\ b & \mbox{if } j>k(c) \end{cases}.
\end{equation}

\subsection{A structure theorem for the class of PBS decompositions} 

The following lemmas are vital in proving many of the succeeding results.

\begin{lemma}\label{prop3000}
	Let $\mathscr D: \mathscr N = \mathscr N_1\cup \cdots\cup \mathscr N_k$ be a \textnormal{PBS} decomposition. If $y'-x_1-\cdots-x_q-y''$ is a path in $\mathscr N$ such that $y'$ and $y''$ are common complexes while $x_1,\cdots,x_q$ are not, then there is a subnetwork $\mathscr N_i$ where $y'$ and $y''$ are connected.
\end{lemma}

\begin{proof}

Suppose $x_1$ is a complex in subnetwork $\mathscr N_i$. By definition, $x_1$ cannot be a complex in any other subnetwork. Hence, $y'-x_1$ is a subpath in $\mathscr N_i$. If $x_2$ is not a complex in $\mathscr N_i$, then the subpath $x_1-x_2$ cannot occur in $\mathscr N_i$ violating the assumption on $x_1$. So, $x_2$ is also a complex in $\mathscr N_i$. Similarly, it can be shown that $x_3, \cdots, x_q$ are also complexes in $\mathscr N_i$. Moreover, $y'-x_1-\cdots-x_q-y''$ is a path in $\mathscr N_i$ indicating that $y'$ and $y''$ are connected in $\mathscr N_i$. 
\end{proof}

\begin{lemma}\label{prop1.700}
	Let $\mathscr D: \mathscr N = \mathscr N_1\cup \cdots\cup \mathscr N_k$ be a \textnormal{PBS} decomposition. If $P:y-x_1-\cdots-x_q$ is a path in $\mathscr N$ such that $y$ is a common complex while $x_1,\cdots,x_q$ are not, then there is exactly one subnetwork $\mathscr N_i$ that contains $P$. 
\end{lemma}

\begin{proof}
(Similar to the proof of Lemma \ref{prop3000}.)
\end{proof}

Let a linkage class  that has a common complex be called $\mathscr C_{\mathscr D}$-\textbf{linkage class}. We remark that Lemma \ref{prop1.700} implies that every non-common complex of a $\mathscr C_{\mathscr D}$-linkage class of $\mathscr N$ belongs to a $\mathscr C_{\mathscr D}$-linkage class of some subnetwork $\mathscr N_i$. 

The following result serves as the {structure theorem for PBS decompositions} which is a generalization of statement (\textit{ii}) of the Structure Theorem for $\mathscr C^*$-decomposition (Theorem \ref{str_thm_for_C*}).

\begin{theorem}\label{str_thm_for_PMM}
	Let $\mathscr D:\mathscr N = \mathscr N_1 \cup \cdots\cup \mathscr N_k$ be a \textnormal{PBS} decomposition such that at least one linkage class of $\mathscr N$ does not contain a common complex. Let $\mathscr N'$ and $\mathscr N'_i$ be the networks obtained after removing from $\mathscr N$ and  $\mathscr N_i$, respectively, the $\mathscr C_{\mathscr D}$-linkage classes. Then, $\mathscr D':\mathscr N' = \mathscr N'_1\cup \cdots\cup \mathscr N'_k$ is a $\mathscr C$-decomposition if $\mathscr D'$ is non-trivial.
\end{theorem}

The condition that $\mathscr N$ has a linkage class that has no common complex ensures that the decomposition obtained after the removal of the $\mathscr C_{\mathscr D}$-linkage classes is not going to be an empty decomposition. The worst possible scenario is that we get a trivial decomposition. We now give the proof of Theorem \ref{str_thm_for_PMM}.

\begin{proof}[Proof \textnormal{of Theorem \ref{str_thm_for_PMM}}]
    The length of $\mathscr D'$ is at least 2 since it is non-empty and non-trivial. Suppose $\mathscr N'=(\mathscr S',\mathscr C', \mathscr R')$ and $\mathscr N'_i=(\mathscr S'_i,\mathscr C'_i, \mathscr R'_i)$. First, we show that $\mathscr D'$ is a decomposition, i.e., $\{\mathscr R'_i\}$ is a partition of $\mathscr R'$. Clearly, $\mathscr R'_i$'s are pairwise disjoint. Let $R\in \mathscr R'$. Then, $R$ is not a reaction in any $\mathscr C_{\mathscr D}$-linkage class of $\mathscr N$. Since $R\in \mathscr R$ and  $\mathscr D$ is a decomposition, $R$ must be a reaction in exactly one subnetwork $\mathscr N_j$. If $R$ belongs to a $\mathscr C_{\mathscr D}$-linkage class of $\mathscr N_j$, then it also belongs to a $\mathscr C_{\mathscr D}$-linkage class of $\mathscr N$ which is a contradiction. Hence, $R$ occurs in $\mathscr N'_j$, i.e., $R\in \mathscr R_j'$ which further implies that $R\in \cup\; \mathscr R'_i$. Thus, $\mathscr R'\subseteq \cup\; \mathscr R'_i$.
	
	Now, we show that $\cup\; \mathscr R'_i \subseteq  \mathscr R'$. Let $R^*:x_1\rightarrow x_2$ be a reaction that is not contained in $\mathscr R'$, i.e., $R^*$ belongs to a $\mathscr C_{\mathscr D}$-linkage class of $\mathscr N$. We want to show that $R^*$ is not contained in any $\mathscr R'_i$. If any of $x_1$ and $x_2$ are common complexes, we are done. Suppose none of them are  common complexes. Let $y$ be the nearest common complex from $x_2$. Let $P:x_1-x_2-z_1-z_2-\cdots-z_p-y$, where $z_1, z_2, \cdots, z_p\notin \mathscr C_{\mathscr D}$, be the shortest path from $x_1$ to $y$. Suppose $R^*$ occurs in subnetwork $\mathscr N_h$. Then, $x_1$ and $x_2$ also occur in $\mathscr N_h$. By Lemma \ref{prop1.700}, $P$ is a path in $\mathscr N_h$ which means that $R^*$ belongs to a $\mathscr C_{\mathscr D}$-linkage class of $\mathscr N_h$. So, $R^*\notin \mathscr R'_h$ or more generally, $R^*$ is not contained in any $\mathscr R'_i$. Hence, $\mathscr D'$ is a decomposition.
	
	Finally, we show that for $i\neq j$, $\mathscr C'_i$ and $\mathscr C'_j$ are disjoint. Note that $\mathscr C'_i \subseteq \mathscr C_i$ and $\mathscr C'_j \subseteq \mathscr C_j$. If $\mathscr C_i \cap \mathscr C_j=\varnothing$, we are done. Suppose $\mathscr C_i \cap \mathscr C_j$ is non-empty. By definition, $\mathscr C'_i \subseteq \mathscr C_i\setminus \mathscr C_{\mathscr D}$ and $\mathscr C'_j \subseteq \mathscr C_j\setminus \mathscr C_{\mathscr D}$. So, $\mathscr C'_i \cap \mathscr C'_j\subseteq (\mathscr C_i\setminus \mathscr C_{\mathscr D}) \cap (\mathscr C_j\setminus \mathscr C_{\mathscr D}) = \varnothing$. Therefore, $\mathscr D'$ is a $\mathscr C$-decomposition.
\end{proof}

\begin{example} \label{haha} Consider the decomposition $\mathscr N=\mathscr N_1\cup \mathscr N_2\cup \mathscr N_3 \cup \mathscr N_4$ with the following reaction graphs.

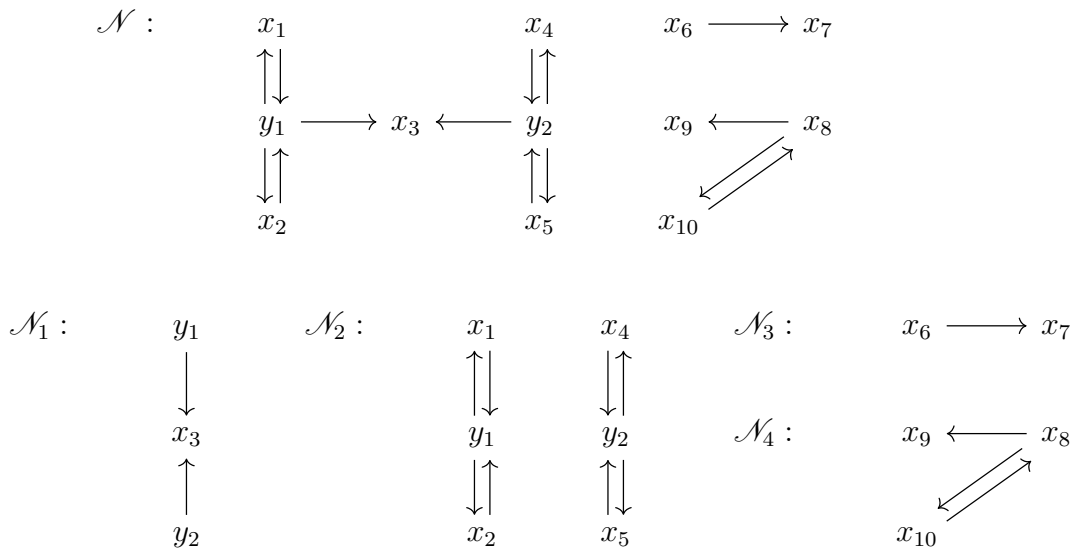
\begin{equation*}
\begin{tikzcd}
\mathscr N:   & x_1 \arrow[d, shift left]                                  &               & x_4 \arrow[d, shift right]                                 & x_6 \arrow[r]                                    & x_7                                   &                                &                                       \\
              & y_1 \arrow[u, shift left] \arrow[d, shift right] \arrow[r] & x_3           & y_2 \arrow[l] \arrow[u, shift right] \arrow[d, shift left] & x_9                                              & x_8 \arrow[l] \arrow[ld, shift right] &                                &                                       \\
              & x_2 \arrow[u, shift right]                                 &               & x_5 \arrow[u, shift left]                                  & x_{10} \arrow[ru, shift right]                   &                                       &                                &
\end{tikzcd}
\end{equation*}

\begin{equation*}
\begin{tikzcd}
\mathscr N_1: & y_1 \arrow[d]                                              & \mathscr N_2: & x_1 \arrow[d, shift left]                                  & x_4 \arrow[d, shift right]                       & \mathscr N_3:                         & x_6 \arrow[r]                  & x_7                                   \\
              & x_3                                                        &               & y_1 \arrow[u, shift left] \arrow[d, shift right]           & y_2 \arrow[u, shift right] \arrow[d, shift left] & \mathscr N_4:                         & x_9                            & x_8 \arrow[l] \arrow[ld, shift right] \\
              & y_2 \arrow[u]                                              &               & x_2 \arrow[u, shift right]                                 & x_5 \arrow[u, shift left]                        &                                       & x_{10} \arrow[ru, shift right] &                                      
\end{tikzcd}
\label{cecil33}
\end{equation*}

\noindent The decomposition is a $\text{PBS}_2(0,2)$ with $\mathscr C_{\mathscr D}=\{y_1,y_2\}$ and $k(c)=2$ (see $\mathscr N_1$ and $\mathscr N_2$). $\mathscr N$ has two linkage classes that do not contain any common complex. Removing the $\mathscr C_{\mathscr D}$-linkage classes from the parent network and the subnetworks yields the decomposition $\mathscr N'=\mathscr N'_3\cup \mathscr N'_4$ given by the following. 
 
\begin{equation}
\begin{tikzcd}
\mathscr N': & x_6 \arrow[r]                  & x_7                                   & \mathscr N'_3: & x_6 \arrow[d] & \mathscr N'_4 & x_9                                  \\
             & x_9                            & x_8 \arrow[l] \arrow[ld, shift right] &                & x_7           &               & x_8 \arrow[u] \arrow[d, shift right] \\
             & x_{10} \arrow[ru, shift right] &                                       &                &               &               & x_{10} \arrow[u, shift right]       
\end{tikzcd} 
\label{cecil34}
\end{equation}

\noindent This decomposition is a $\mathscr C$-decomposition.
\end{example}

\subsection{The $\mathscr C_{c,d}$ subclasses of the PBS class}

\begin{definition}
For a nonnegative integer $d$, the class of pairwise binary-sized decompositions $\mathscr {PBS}_d$ is the union of all sets $\text{PBS}_d(b,c)$ with $0\leq b\leq c \leq d$.
\end{definition} 

\noindent We are interested in identifying subclasses of $\mathscr {PBS}_d$ which possess similar properties with respect to incidence independence deriving from the structure of their sets of common complexes. 

\begin{example}
$\mathscr C^*$-decompositions are contained in the subclass $\text{PBS}_0\cup \text{PBS}_1(0,1)\cup \text{PBS}_1(1,1)$. All decompositions in this subclass are incidence independent as will be shown in the next section. The $\mathscr C^*$-decompositions consist of the special case when the single $\mathscr C_{\mathscr D}$ element $y$ is the $0$ complex.
\end{example}

The above example motivates the following definition of subclasses with $b = 0$ or $b=c$.

\begin{definition}
The subclass $\mathscr C_{c,d}$ is the union $\text{PBS}_0\cup \text{PBS}_d(0,c)\cup \text{PBS}_d(c,c)$.    
\end{definition}

\begin{example}
The subclass $\mathscr C_{0,0}$  is the class of $\mathscr C$-decompositions.
\end{example}

\begin{example}
The subclass $\mathscr C_{1,1}$ contains the class of $\mathscr C^*$-decompositions as well as the decomposition of Schmitz's carbon cycle subnetwork given in Example \ref{carbon cycle model}.
\end{example}

\section{The subclass $\mathscr C_{d,d}$ of pairwise min-max decompositions}

In this section, we we look into the subclass $\mathscr C_{d,d}$ of PBS decompositions which we call class of pairwise min-max decompositions. 

\subsection{The CCF for pairwise min-max decompositions}

We define the set of pairwise min-max decompositions as follows.

\begin{definition}{\label{def1.3}}
	The set of \textbf{pairwise min-max} ({PMM}) \textbf{decompositions} is defined for a set of common complexes $\mathscr C_{\mathscr D}$ with $d$ elements by either $\mathscr C_i \cap \mathscr C_j = \varnothing$ (minimal value) or $\mathscr C_i \cap \mathscr C_j = \mathscr C_{\mathscr D}$ (maximal value). 
\end{definition}

The PIM of a PMM decomposition is the matrix $[m_{ij}]$ where
\begin{equation}
    m_{ij} = \begin{cases} d & \mbox{if } i< j \leq k(c)\\ 0 & \mbox{if } j>k(c) \end{cases}.
\end{equation}

A specific case of Proposition \ref{3C_cri_PBS} gives a 3C-criterion for the incidence independence of the subclass $\text{PBS}_d(0,d)\cup \text{PBS}_d(d,d)$: 
\begin{equation}
    \displaystyle \sum \ell_i - \ell = d(k(c)-1).
\end{equation}

\subsection{A fundamental property of PMM decompositions}

The following proposition provides a sufficient condition for the incidence independence of a PMM decomposition where each of the CRN's linkage classes contains at most one common complex. This result allows us to determine the incidence independence of a PMM decomposition by just observing the structure of the CRN's reaction graph.

\begin{proposition}{\label{prop7}}
	Let $\mathscr D: \mathscr N = \mathscr N_1\cup \cdots\cup \mathscr N_k$ be a \textnormal{PMM} decomposition such that every linkage class of $\mathscr N$ contains at most one common complex. Then, $\mathscr D$ is incidence independent.
\end{proposition}

\begin{proof}
    Let $\mathscr N_1, \mathscr N_2, \cdots, \mathscr N_{k(c)}$ be the subnetworks containing $\mathscr C_{\mathscr D}$. Also, let $\mathscr N'$ and $\mathscr N'_i$ be the networks obtained after removing from $\mathscr N$ and  $\mathscr N_i$, respectively, the $\mathscr C_{\mathscr D}$-linkage classes. Let $\mathscr D': \mathscr N' = \mathscr N'_1\cup \cdots\cup \mathscr N'_k$. Clearly, $\mathscr N'$ has $\ell-d$ number of linkage classes. On the other hand, Lemma \ref{prop1.700} suggests that $\mathscr N'_i$ has $\ell_i-d$ number of linkage classes for $i=1,\cdots,k(c)$. 
    
    Suppose $\mathscr D'$ is non-empty. Then, $\mathscr D'$ is either a trivial decomposition or a $\mathscr C$-decomposition (Theorem \ref{str_thm_for_PMM}). The former immediately implies that 
        \begin{equation}\label{okkk}
       \displaystyle \sum^{k(c)}_{i=1} (\ell_i-d)+\sum^k_{i=k(c)+1}\ell_i= \ell-d.
    \end{equation}
    \noindent Note that (\ref{okkk}) still holds if $\mathscr D'$ is a $\mathscr C$-decomposition according to Theorem \ref{str_thm_for_C}. Then,
    \begin{equation}\label{basta}
    \displaystyle \sum^{k(c)}_{i=1}\ell_i-d(k(c))+\sum^k_{i=k(c)+1}\ell_i=\ell-d \Leftrightarrow \displaystyle \sum^k_{i=1}\ell_i-\ell=d(k(c)-1)    
    \end{equation}

    Now, suppose $\mathscr D'$ is empty. Then, it means that every linkage class in $\mathscr N$ and $\mathscr N_i$'s is $\mathscr C_{\mathscr D}$-linkage class. It follows that
    \begin{equation}\label{halahala}
       \displaystyle \sum^{k}_{i=1} (\ell_i-d)=0 \text{\;and\;} \ell-d=0 \Rightarrow \sum^{k}_{i=1} (\ell_i-d)=\ell-d \Leftrightarrow \displaystyle \sum^k_{i=1}\ell_i-\ell=d(k-1).
    \end{equation}
    
\noindent By Proposition \ref{3C_cri_PBS}, (\ref{basta}) and (\ref{halahala}) suggest that, in any case, $\mathscr D$  is incidence independent. 

\end{proof}

It immediately follows from Proposition \ref{prop7} that PMM decompositions with $|\mathscr C_{\mathscr D}|\leq 1$ behave in the same way as $\mathscr C^*$-decompositions and are all incidence independent. 

\begin{example}
	The decomposition of the subnetwork of the of the Schmitz's model given in Example \ref{carbon cycle model} is a PMM decomposition with $\mathscr C_{\mathscr D}=\{M_1\}$ which means that, indeed, it is incidence independent.    
	
\end{example}

\begin{example} Consider $\mathscr N$ whose reaction graph is given by
 
    
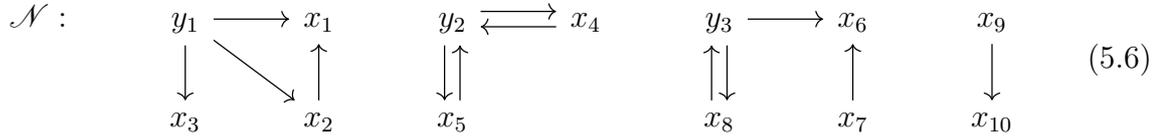
\begin{equation}
\begin{tikzcd}
\mathscr N:   & y_1 \arrow[r] \arrow[rd] \arrow[d] & x_1           & y_2 \arrow[r, shift left] \arrow[d, shift right] & x_4 \arrow[l, shift left]  & y_3 \arrow[r] \arrow[d, shift left] & x_6                        & x_9 \arrow[d] \\
              & x_3                                & x_2 \arrow[u] & x_5 \arrow[u, shift right]                       &                            & x_8 \arrow[u, shift left]           & x_7 \arrow[u]              & x_{10}                  
\end{tikzcd}
\label{cecil35}
\end{equation}

\noindent $\mathscr N$ can be decomposed as follows.
 
    
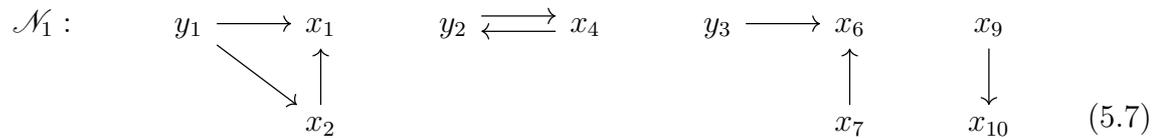
\begin{equation}
\begin{tikzcd}
\mathscr N_1: & y_1 \arrow[r] \arrow[rd]           & x_1           & y_2 \arrow[r, shift left]                        & x_4 \arrow[l, shift left]  & y_3 \arrow[r]                       & x_6                        & x_9 \arrow[d] \\
              &                                    & x_2 \arrow[u] &                                                  &                            &                                     & x_7 \arrow[u]              & x_{10}        \\
\mathscr N_2: & y_1 \arrow[r]                      & x_3           & y_2 \arrow[r, shift right]                       & x_5 \arrow[l, shift right] & y_3 \arrow[r, shift right]          & x_8 \arrow[l, shift right] &              
\end{tikzcd}
\label{cecil36}
\end{equation}

\noindent $\mathscr N=\mathscr N_1 \cup \mathscr N_2$ is a PMM decomposition with $\mathscr C_{\mathscr D}=\{y_1,y_2,y_3\}$. Notice that each linkage class of $\mathscr N$ contains at most one element from $\mathscr C_{\mathscr D}$. Hence, by Proposition \ref{prop7}, this decomposition is incidence independent.
\end{example}

\subsection{Special properties of the set $\text{PBS}_d(d,d)$ of $\mathscr C_{\mathscr D}$-decompositions}

We now discuss a PBS decomposition class where every subnetwork contains the whole $\mathscr C_{\mathscr D}$.

\begin{definition}{\label{def1.2}}
	For any subset $\mathscr C_{\mathscr D}$ of $\mathscr C$ with $d$ elements, the $\mathscr C_{\mathscr D}$-\textbf{decomposition} is defined by $\mathscr C_i \cap \mathscr C_j = \mathscr C_{\mathscr D}$.
\end{definition}

\noindent A special case of Proposition \ref{3C_cri_PBS} gives $\displaystyle \sum \ell_i - \ell =d(k-1)$ - the 3C-criterion for incidence independence of $\mathscr C_{\mathscr D}$-decompositions. 

The incidence independence of a $\mathscr C_{\mathscr D}$-decomposition where every linkage class of the parent network contains at most one common complex is guaranteed by Proposition \ref{prop7}. We now investigate the case where a linkage class of $\mathscr N$ contains more than one common complex. 

\begin{proposition}\label{y1_y2}
	Let $\mathscr D: \mathscr N = \mathscr N_1\cup \cdots\cup \mathscr N_k$ be a $\mathscr C_{\mathscr D}$-decomposition with $\mathscr C_{\mathscr D}=\{y_1,y_2\}$ and $\ell>2$. If there is at most one path in $\mathscr N$ that connects  $y_1$ and $y_2$, then  $\mathscr D$ is incidence independent.
\end{proposition}

\begin{proof}
	If there is no path in $\mathscr N$ connecting $y_1$ and $y_2$, then $y_1$ and $y_2$ must be in separate linkage classes. So, By Proposition \ref{prop7},  $\mathscr D$ is incidence independent. 
	
	Suppose that in $\mathscr N$, there is exactly one path from $y_1$ to $y_2$. If such path is either of the reactions $y_1\rightarrow y_2$ or $y_2\rightarrow y_1$, then it will be contained in exactly one subnetwork. Now, if the path is given by $P: y_1-a_1-\cdots-a_q-y_2$ where $a_1,\cdots,a_q$ are not common complexes. Then, by Lemma \ref{prop3000}, there will be a subnetwork $\mathscr N_j$ that contains $P$. The uniqueness of $P$ indicates that no other subnetwork will contain it, i.e., in any other subnetwork, $y_1$ and $y_2$ are in separate linkage classes.
	
	Let $\mathscr N'$ and $\mathscr N'_i$ be the networks obtained after removing from $\mathscr N$ and  $\mathscr N_i$, respectively, the linkage classes that contain $y_1$ or $y_2$. Let $\mathscr D': \mathscr N' = \mathscr N'_1\cup \cdots\cup \mathscr N'_k$. From Lemma \ref{prop1.700}, $\mathscr N'_j$ has $\ell_j-1$ number of linkage classes while every other subnetwork has $\ell_i-2$. If $\mathscr D'$ is trivial, we have
	\begin{equation}\label{hold}
	    (\ell_j - 1) + \displaystyle \sum_{i\neq j} (\ell_i-2)=\ell-1
	\end{equation}
	
	If $\mathscr D'$ is non-trivial, by Theorem \ref{str_thm_for_PMM}, it is a $\mathscr C$-decomposition and (\ref{hold}) holds. In any case, $\displaystyle \sum^k_{i=1}\ell_i - \ell = 2(k-1)$ which means that by Proposition \ref{3C_cri_PBS}, $\mathscr D$ is incidence independent.
\end{proof}

	For common complexes $y'$ and $y''$, let $P_{y'}^{y''}$ denotes a path from $y'$ to $y''$ whose ``internal complexes'' are non-common complexes. $P_{y'}^{y''}$ is the path $y'-y''$ if there are no internal complexes. A generalization of Proposition \ref{y1_y2} is given in the following result.

\begin{theorem}\label{prop2.6}
	Let $\mathscr D: \mathscr N = \mathscr N_1\cup \cdots\cup \mathscr N_k$ be a $\mathscr C_{\mathscr D}$-decomposition such that $\mathscr C_{\mathscr D}=\{y_1,y_2,\cdots,y_d\}$ and $\ell>1$. If there is exactly one path in $\mathscr N$ that connects  $y_1,y_2, \cdots, y_d$ say $P_{y_1}^{y_2}-\cdots-P_{y_{d-1}}^{y_{d}}$ such that $P_{y_i}^{y_{i+1}}$ contains at least one non-common complex, then $\mathscr D$ is incidence independent.
\end{theorem}

\begin{proof}
We prove this by induction on the number $d$ of elements of $\mathscr C_{\mathscr D}$. The case where $d=2$ is given in Proposition \ref{y1_y2}. For the hypothesis of induction, we assume that every $\mathscr C_{\mathscr D}$-decomposition of a CRN with $|\mathscr C_{\mathscr D}|= d$ such that the common complexes are connected in the network by a unique path (having a property described in the proposition) is incidence independent. 

Let $\mathscr D': \mathscr N' = \mathscr N'_1\cup \cdots\cup \mathscr N'_{k'}$ be a $\mathscr C_{\mathscr D}$-decomposition with $\mathscr C_{\mathscr D}=\{y_1,y_2,\cdots, y_{d+1}\}$. Suppose that in $\mathscr N'$, there is exactly one path that connects $y_1,y_2,\cdots,y_{d+1}$ say $P_{y_1}^{y_2}-\cdots-P_{y_{d-1}}^{y_{d}}-P_{y_d}^{y_{d+1}}$ with the property that $P_{y_i}^{y_{i+1}}$ contains at least one non-common complex. Denote by $P({y_{d+1}})$ a path with complexes consist only of $y_{d+1}$ and non-common complexes, i.e., a path that passes through $y_{d+1}$ whose other complexes are non-common complexes. Some of the paths $P({y_{d+1}})$ (see the broken segments that passes through ${y_{d+1}}$) are shown in (\ref{eqn2.1}). 
 
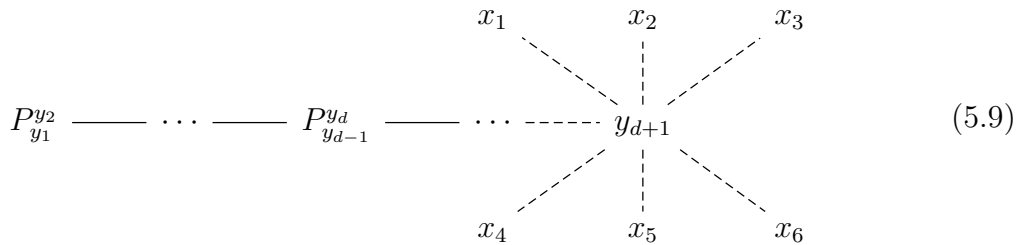
\begin{equation}
\label{eqn2.1}
\begin{tikzcd}
                                 &                           &                                        & x_1                               & x_2                                                                                                                                                                           & x_3 \\
P_{y_1}^{y_2} \arrow[r, no head] & \cdots \arrow[r, no head] & P_{y_{d-1}}^{y_{d}} \arrow[r, no head] & \cdots \arrow[r, no head, dashed] & y_{d+1} \arrow[u, no head, dashed] \arrow[d, no head, dashed] \arrow[lu, no head, dashed] \arrow[ld, no head, dashed] \arrow[ru, no head, dashed] \arrow[rd, no head, dashed] &     \\
                                 &                           &                                        & x_4                               & x_5                                                                                                                                                                           & x_6
\end{tikzcd}
\end{equation}

Let $\mathscr N^*=(\mathscr S^*,\mathscr C^*,\mathscr R^*)$ and $\mathscr N_i^*=(\mathscr S^*_i,\mathscr C^*_i,\mathscr R^*_i)$ be the networks obtained after removing from $\mathscr N'$ and $\mathscr N'_i$ every path $P(y_{d+1})$. We now claim that $\mathscr D^*:\mathscr N^* = \mathscr N^*_1\cup \cdots\cup \mathscr N^*_{k'}$ is a $\mathscr C_{\mathscr D}$-decomposition with $\mathscr C_{\mathscr D}=\{y_1,y_2,\cdots,y_d\}$. Let $Q$ be the set of all paths $P(y_{d+1})$ and for convenience, let us call a path contained in $Q$ a $Q$-path. First, we show that $\{\mathscr R^*_i\}$ is a partition of $\mathscr R^*$. Clearly, $\mathscr R^*_i$'s are pairwise disjoint. Consider a reaction $R\notin \cup \mathscr R_i^*$. Then, $R$ is ``contained'' in a $Q$-path in some subnetwork $\mathscr N'_i$. It follows that $R$ is also contained in a $Q$-path in $\mathscr N'$. Hence, $R\notin \mathscr R^*$. 

This time, let $R:x'\rightarrow x''\notin \mathscr R^*$. Then, $R$ is in a $Q$-path in $\mathscr N$. Consider the path $x'-x''-a_1-\cdots-a_q-y_{d+1}$. From Lemma \ref{prop1.700}, this path occurs in some subnetwork $\mathscr N'_j$, i.e., $R$ is in a $Q$-path in $\mathscr N'_j$. So, $R\notin \mathscr R_j^*$ and in general, $R \notin \cup \mathscr R_i^*$. This means that $\mathscr D^*$ is a decomposition. 

Clearly, $y_1,y_2,\cdots,y_{d-1}$ appear as complexes in every subnetwork $\mathscr N^*_i$. By virtue of Lemma \ref{prop1.700}, $Q$-paths are contained in the same linkage class of every subnetwork $\mathscr N'_i$. Now, due to uniqueness and Lemma \ref{prop3000}, the path $P_{y_d}^{y_{d+1}}$ occurs in exactly one subnetwork say $\mathscr N'_g$. This means that in $\mathscr N'_g$, every $Q$-path is contained in the linkage class that also contains $y_d$. On the other hand, in every other subnetwork $\mathscr N'_i$, the linkage class that contains $y_{d+1}$ does not have any other common complex. Since there is at least one non-common complex in $P_{y_d}^{y_{d+1}}$, there is at least one reaction left in the linkage class of $\mathscr N^*_g$ that contains $y_d$. Hence, $y_d$ also appears as a complex in each $\mathscr N^*_i$. Thus, $y_{d+1}$ is the only common complex that was removed to obtain $\mathscr N^*$ and $\mathscr N_i^*$, So, $\mathscr D^*$ is a $\mathscr C_{\mathscr D}$-decomposition.

Next, we show that there is exactly one path connecting $y_1,y_2,\cdots,y_d$ in $\mathscr N^*$. Suppose that such path does not exist. Then, there must be a non-common complex $w$ in some path $P_{y_{i-1}}^{y_{i}}$ that was removed from $\mathscr N'$ in obtaining $\mathscr N^*$, i.e.,  $w$ is complex in some $Q$-path in $\mathscr N'$.
 
\begin{equation}
\begin{tikzcd}
                       &                            &                                                              &                        & Q\textnormal{-path} \arrow[rd, no head, dashed, bend left] &         \\
y_1 \arrow[r, no head] & y_{i-1} \arrow[r, no head] & w \arrow[r, no head] \arrow[rru, no head, dashed, bend left] & y_i \arrow[r, no head] & y_d \arrow[r, no head]                                     & y_{d+1}
\end{tikzcd}
\label{cecil37}
\end{equation}

\noindent This implies that $P_{y_1}^{y_2}-\cdots -P_{y_{i-2}}^{y_{i-1}}-\cdots-w-\cdots -P_{y_{d+1}}^{y_{d}}-\cdots-P_{y_{i+1}}^{y_{i}}$ is another path connecting the $y_1,y_2,\cdots,y_{d+1}$ which is a contradiction. From the induction hypothesis, we conclude that $\mathscr N^* = \mathscr N^*_1\cup \mathscr N^*_2\cup \cdots\cup \mathscr N^*_{k'}$ is incident-independent. From Proposition \ref{3C_cri_PBS}, we have 
\begin{equation}\label{pogi}
    \displaystyle \sum_{i=1}^{k'} \ell_i^*-\ell^*=d(k'-1) 
\end{equation}
\noindent where $\ell^*$ and $\ell_i^*$ are the number of linkage classes of $\mathscr N^*$ and $\mathscr N_i^*$, respectively.

Let $\ell'$ and $\ell'_i$ be the number of linkage classes of $\mathscr N'$ and $\mathscr N'_i$, respectively. Then, we have $\ell'_g=\ell^*_g$ and $\ell'_i=\ell_i^*+1$ for $i\neq g$. Also, we have $\ell'=\ell^*$. Thus,
\begin{equation*}
    \sum_{i=1}^{k'} \ell'_i-\ell'= \ell'_g + \sum_{i\neq g}\ell'_i - \ell'=\ell^*_g+ \sum_{i\neq g}(\ell_i^*+1)-\ell^*=\sum_{i=1}^{k'} \ell_i^*-\ell^*+(k'-1).
\end{equation*}
\noindent From (\ref{pogi}), we have
\begin{equation*}
    \sum_{i=1}^{k'} \ell'_i-\ell'=d(k'-1)+(k'-1)=(d+1)(k'-1).
\end{equation*}
\noindent Therefore, from Proposition \ref{3C_cri_PBS}, $\mathscr D'$ is incidence independent. 
\end{proof}

\begin{example}
Consider decomposition $\mathscr N=\mathscr N_1\cup \mathscr N_2\cup \mathscr N_3 \cup \mathscr N_4$ in Example \ref{haha}. There is only the path $y_1-x_3-y_2$ connecting the common complexes $y_1$ and $y_2$ in $\mathscr N$ such that the internal complex $x_3$ is a non-common complex. By Theorem \ref{prop2.6}, the decomposition is incidence independent.        
\end{example}

As shown in the next result, the incidence independence of a $\mathscr C_{\mathscr D}$-decomposition can still be achieved even if each linkage class of the parent network has none or has exactly two common complexes.

\begin{proposition}\label{prop2.7}
	Let $\mathscr D: \mathscr N = \mathscr N_1\cup \cdots\cup \mathscr N_k$ be a $\mathscr C_{\mathscr D}$-decomposition with $\dfrac{d}{2}<\ell$. Suppose that each linkage class of $\mathscr N$ either contains zero or exactly two common complexes that are connected by a unique path. Then, $\mathscr D$ is incidence independent.
\end{proposition}

\begin{proof}
	Let $\mathscr L_1, \cdots, \mathscr L_{d/2}$ be the $\mathscr C_{\mathscr D}$-linkage classes of $\mathscr N$. For linkage class $\mathscr L_j$, denote by $y_j^1$ and $y_j^2$ the common complexes it posses. From the assumption and Lemma \ref{prop3000}, $y_j^1$ and $y_j^2$ belong to the same linkage class in some subnetwork say $\mathscr N_h$ but are separated in every other subnetworks. 	
	
	Let $\mathscr N'$ and $\mathscr N'_i$ be the networks obtained after removing from $\mathscr N$ and  $\mathscr N_i$, respectively, the $\mathscr C_{\mathscr D}$-linkage classes. Let $\mathscr D': \mathscr N' = \mathscr N'_1\cup \cdots\cup \mathscr N'_k$. If $\mathscr D'$ is non-trivial, it is a $\mathscr C$-decomposition by Theorem \ref{str_thm_for_PMM}. Let $\ell'$ and $\ell_i'$ be the number of linkage classes of $\mathscr N'$ and $\mathscr N_i'$ and $c_i$ be the number of pairs of connected $y_j^1$, $y_j^2$ in $\mathscr N_i$. Thus, $\ell_i'=\ell_i-c_i-2\bigg(\dfrac{d}{2}-c_i\bigg)$. By Theorem \ref{str_thm_for_C}, we have
	\begin{equation}
	    \displaystyle \sum_{i=1}^k \ell'_i = \ell' \Leftrightarrow \; \displaystyle \sum_{i=1}^k \bigg[\ell_i-c_i-2\bigg(\dfrac{d}{2}-c_i\bigg)\bigg] = \ell-\dfrac{d}{2} \Leftrightarrow \; \displaystyle \sum_{i=1}^k \ell_i-\ell=2k\bigg(\dfrac{d}{2}\bigg)-\dfrac{d}{2}-\displaystyle \sum_{i=1}^k c_i.
	\end{equation}

	\noindent Since $\displaystyle \sum_{i=1}^k c_i=\dfrac{d}{2}$, we get
	\begin{equation}\label{pugo}
		\sum_{i=1}^k \ell_i-\ell=2k\bigg(\dfrac{d}{2}\bigg)-2\bigg(\dfrac{d}{2}\bigg)=d(k-1).
	\end{equation}
	\noindent It can be shown that we still get (\ref{pugo}) even if $\mathscr D'$ is trivial. In any case, $\mathscr D$ is incidence independent by Proposition \ref{3C_cri_PBS}.
\end{proof}

\section{A sufficient condition for complex balancing of PLK systems with incidence independent and complex balanced decompositions}

In this section, we turn to the question: for an incidence independent and complex balanced decomposition (i.e., each subnetwork is complex balanced) of a kinetic system, which additional conditions must hold to ensure that the parent system is complex balanced? Recall that Farinas et al. \cite{fari2021} established the surprisingly general result that the conclusion holds for any $\mathscr C$-decomposition, i.e., the additional purely structural condition of the set of common complexes $\mathscr C_{\mathscr D}=\varnothing$ and any kinetic system (see the following theorem). 

\begin{theorem}[Proposition 10, \cite{fari2021}]\label{lahat}
	Let $(\mathscr N, K)$ be a \textnormal{CKS} and $\mathscr N=\mathscr N_1\cup \cdots \cup \mathscr N_k$ be a weakly reversible $\mathscr C$-decomposition $($i.e., a $\mathscr C$-decomposition where each subnetwork is weakly reversible$)$. If $Z_+(\mathscr N_i, K_i)\neq \varnothing$ for each subnetwork, then $Z_+(\mathscr N, K)\neq \varnothing$. 
\end{theorem}

\noindent In  section \ref{counter_example}, we show that this result is singular in its purely structural character: as soon as $\mathscr C_{\mathscr D}\neq \varnothing$, there are going to be kinetic systems with an incidence independent and complex balanced decomposition that are not complex balanced.  Accordingly, our new result, while imposing no conditions on the set of common complexes of the incidence independent decomposition, is valid only for a class of power law kinetic systems. After brief reviews of relevant results in sections \ref{LP_type} and \ref{induced_decompositions}, we state and prove the new result in section \ref{main_result}. We then briefly compare it with the generalization of the Deficiency Zero Theorem of Fortun et al. \cite{fort2019} and sketch some interesting research questions.

\subsection{A simple counterexample for $\mathscr C_{\mathscr D}\neq \varnothing$}\label{counter_example}

Consider the $\mathscr C^*$-decomposition $\mathscr N=\mathscr N_1\cup \mathscr N_2$ given by the following.
 
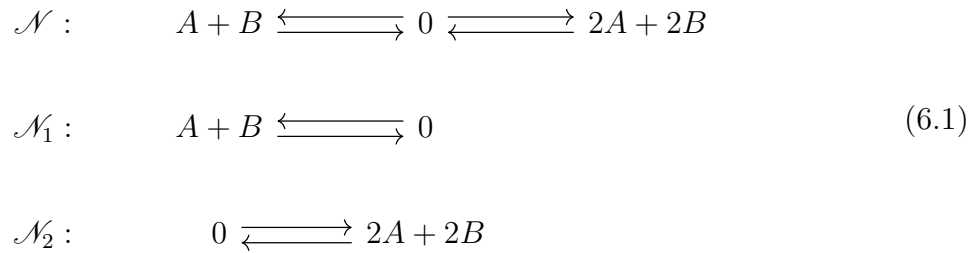
\begin{equation}
	\label{eqn3.11}
	\begin{tikzcd}
		\mathscr N:    & A+B \arrow[r, shift right] & 0 \arrow[l, shift right] \arrow[r, shift left] & 2A+2B \arrow[l, shift left] \\
		\mathscr N_1:  & A+B \arrow[r, shift right] & 0 \arrow[l, shift right]                       &                             \\
		\mathscr N_2:  & 0 \arrow[r, shift left]    & 2A+2B \arrow[l, shift left]                    &                            
	\end{tikzcd}
\end{equation}

\noindent Suppose that $(\mathscr N, K)$ is endowed with mass action kinetics. For the reversible and deficiency zero subsystem $(\mathscr N_1, K_1)$ we have the equations
\begin{equation*}
dA/dt = dB/dt = k_1-k_2AB   
\end{equation*}
\noindent and  $AB=k_1/k_2$ at equilibrium. Similarly, for the subsystem $(\mathscr N_2, K_2)$, which is also reversible and deficiency zero, we get 
\begin{equation*}
dA/dt = dB/dt=2k_3-2k_4A^2B^2   
\end{equation*}
\noindent and $AB=\sqrt{k_3/k_4}$ at equilibrium. For the parent system, we have 
\begin{equation*}
dA/dt = dB/dt = k_1-k_2AB+2k_3-2k_4A^2B^2. 
\end{equation*}
For the rate vector $(2,4,1,1/2)$ we have 
\begin{equation*}
 Z_+(\mathscr N_1, K_1)=\{(A,B)\in \mathbb R^2| AB=1/2\}\textnormal{\;and\;} Z_+(\mathscr N_2, K_2)=\{(A,B)\in \mathbb R^2| AB=\sqrt{2}\}   
\end{equation*}
so that their intersection is empty, suggesting that $(\mathscr N, K)$ is not complex balanced. 
On the other hand, $(\mathscr N, K)$ is complex balanced for the rate vector $(1,1,1,1)$ with 
\begin{equation*}
E_+(\mathscr N, K)=Z_+(\mathscr N, K)=Z_+(\mathscr N_1, K_1)=Z_+(\mathscr N_2, K_2)=\{(A,B)\in \mathbb R^2| AB=1\}.    
\end{equation*}    
Hence, even mass action systems on the same network and the same incidence independent and complex balanced decomposition (in fact a $\mathscr C^*$-decomposition) may or may not be complex balanced.

\subsection{Kinetic systems of LP type}\label{LP_type}

In this and the next subsection, we collect some relevant concepts for the new result. 

\begin{definition}
A kinetic system $(\mathscr N, K)$ is of type PLP (\textbf{positive equilibria log-parametrized}) if
\begin{enumerate}[i.]
    \item $E_+(\mathscr N, K)\neq \varnothing$ and
    \item $E_+(\mathscr N, K)=\{x\in \mathbb R_{>0}^{\mathscr S}| \log x- \log x^* \in (P_E)^{\perp}\}$
\end{enumerate}
where $P_E$ is a subspace of $\mathbb R^{\mathscr S}$ and $x^*$ is a positive equilibrium.
\end{definition}

\begin{definition}
A kinetic system $(\mathscr N, K)$ is of type CLP (\textbf{complex balanced equilibria log-parametrized}) if
\begin{enumerate}[i.]
    \item $Z_+(\mathscr N, K)\neq \varnothing$ and
    \item $Z_+(\mathscr N, K)=\{x\in \mathbb R_{>0}^{\mathscr S}| \log x- \log x^* \in (P_Z)^{\perp}\}$
\end{enumerate}
where $P_Z$ is a subspace of $\mathbb R^{\mathscr S}$ and $x^*$ is a complex balanced equilibrium.
\end{definition}

A kinetic system is bi-LP if it is of PLP and of CLP type and $P_E=P_Z$. We will use the shorter PLP system, CLP system and bi-LP system notation as well as the collective term ``LP systems".

A key property of an LP system was in principle already derived by Feinberg in his 1979 lectures \cite{fein1979} as shown in \cite{magp2021}:

\begin{theorem}
Let $(\mathscr N, K)$ be a chemical kinetic sytem.
\begin{enumerate}[i.]
    \item If $(\mathscr N, K)$ is a \textnormal{PLP} system, then $|E_+(\mathscr N, K)\cap Q|=1$ for any positive coset $Q$ of $P_E$ in $\mathbb R^{\mathscr S}$.
    \item If $(\mathscr N, K)$ is a \textnormal{CLP} system, then $|Z_+(\mathscr N, K)\cap Q|=1$ for any positive coset $Q$ of $P_Z$ in $\mathbb R^{\mathscr S}$.
    \item If $(\mathscr N, K)$ is a bi-\textnormal{LP} system, then it is absolutely complex balanced, i.e., each positive equilibrium is complex balanced.
\end{enumerate}
\end{theorem}

The concepts of kinetic order subspace and kinetic deficiency for cycle terminal PL-RDK systems were introduced by M\"{u}ller and Regensburger in \cite{mure2014} where they also presented their theory of generalized mass action systems (GMAS). Hernandez and Mendoza developed these concepts in \cite{hern2021} by first constructing the reaction network induced by the kinetic complexes which we are adapting in this paper.

Let $\mathscr R(y)$ be the set of (branching) reactions having $y$ as reactant complex.
\begin{definition} Let $(\mathscr N, K)$ be a cycle terminal PLK system. Given a complex $y$, the set of \textbf{kinetic complexes} of $y$ is the set $\widetilde{\mathscr C}(y):=\{F_R\;|\;R\in \mathscr R(y)\}$ ($F_R$ refers to the kinetic order vector that corresponds to reaction $R$). Moreover, for a reaction $R: y\rightarrow y'$, $\widetilde{\mathscr R}(R):=\{\widetilde{y}\rightarrow \widetilde{y'}\;|\;\widetilde{y}\in \widetilde{\mathscr C}(y) \textnormal{\;and\;}  \widetilde{y'}\in \widetilde{\mathscr C}(\widetilde{y'})\}$ gives the set of 
\textbf{kinetic complex reactions} of $R$.
\end{definition}
\noindent The reaction network of kinetic complexes of a cycle terminal PLK system is defined as follows.

\begin{definition} The set of kinetic complexes induced a reaction network given by $\widetilde{\mathscr N}=(\mathscr S, \widetilde{\mathscr C}, \widetilde{\mathscr R})$ where $\widetilde{\mathscr C}=\displaystyle{\bigcup_y}\widetilde{\mathscr C}(y)$ and $\widetilde{\mathscr R}=\displaystyle{\bigcup_y}\widetilde{\mathscr R}(y)$.
\end{definition}

\noindent The orders of $\widetilde{\mathscr C}$ and $\widetilde{\mathscr R}$ are denoted by $\widetilde{n}$ and $\widetilde{p}$, respectively. In addition, the incidence and stoichiometric maps of $\widetilde{\mathscr N}$ are respectively denoted by $\widetilde{I}_a$ and $\widetilde{N}$.

\begin{definition}
The \textbf{kinetic order subspace} $\widetilde{S}$ of $(\mathscr N, K)$ is the defined to be the image of $\widetilde{N}$ and its dimension $\widetilde{s}$ is also referred to as \textbf{kinetic rank}. The \textbf {kinetic complex deficiency} is the nonnegative number given by $\delta_{\widetilde{N}}=\widetilde{n}-\widetilde{\ell}-\widetilde{s}$. On the other hand, the \textbf {kinetic deficiency} is the nonnegative number $\widetilde{\delta}=n-{\ell}-\widetilde{s}$.
\end{definition}

A key ingredient of the new result is the following theorem of M\"{u}ller and Regensburger \cite{mure2014}.

\begin{theorem}\label{muller-regensburger}
Let $(\mathscr N, K)$ be a weakly reversible \textnormal{PL-RDK} system. If $(\mathscr N, K)$ is complex balanced, then it is of \textnormal{CLP} type with $P_Z=\widetilde{S}$. 
\end{theorem}

With the concept of kinetic deficiency, they provide a characterization of complex balancing: if the kinetic deficiency is zero, then any weakly reversible PL-RDK system is complex balanced. If it is positive, they provide a necessary and sufficient condition for the system to be complex balanced.

\subsection{Induced decompositions of kinetic complexes}\label{induced_decompositions}

The second part of the sufficient condition is based on the concepts of an induced subnetwork of kinetic complexes and a corresponding decomposition, which were introduced in \cite{hern2021}. For a PL-RDK system, a kinetic complex is the column of the $T$ matrix assigned to a reactant complex. For a decomposition $\mathscr D: \mathscr N=\mathscr N_1\cup \cdots \cup \mathscr N_k$, we set $\widetilde{\mathscr N}_i:=$ subnetwork of $\widetilde{\mathscr N}$ induced by $\mathscr N_i$, i.e., take the reactions defining it, then form the kinetic complexes (still in the sense of M\"{u}ller-Regensburger).

\begin{definition}
The \textbf{induced subnetwork} $\widetilde{\mathscr N}_{\mathscr D}$ is defined as the union of the $\widetilde{\mathscr N}_i$. If the covering $\widetilde{\mathscr D}$ is a decomposition, we call it the \textbf{induced decomposition}.    
\end{definition}

\noindent We denote by $\widetilde{n}_{\mathscr D}$ and $\widetilde{\ell}_{\mathscr D}$ the number of complexes and linkage classes of $\widetilde{\mathscr N}_{\mathscr D}$, respectively.

In the following, we will assume that the covering is indeed a decomposition (so that in the examples, this has to be verified). Note however, that the following proposition holds:

\begin{proposition}\label{tela}
If, for an induced covering $\widetilde{\mathscr D}$, the flux spaces $\widetilde{S_i}$ form a direct sum, then the covering is an independent decomposition.
\end{proposition}

\begin{example}
Recall the incidence independent decomposition $\mathscr D_1: \mathscr N=\mathscr N_1\cup \mathscr N_2$ given in Example \ref{carbon cycle model} where $\mathscr N$ is the weakly reversible and deficiency zero subnetwork of the Schmitz's pre-industrial carbon cycle model (\ref{graph1000}). In \cite{fort2018}, $\mathscr N$ is assigned with the following kinetic order matrix which makes the system PL-NDK.  
\begin{equation}
F =
  \begin{blockarray}{*{6}{c} l}
    \begin{block}{*{6}{>{$\footnotesize}c<{$}} l}
      $M_1$ & $M_2$ & $M_3$ & $M_4$ & $M_5$ & $M_6$ \\
    \end{block}
    \begin{block}{[*{6}{c}]>{$\footnotesize}l<{$}}
      0 & 0 & 0 & 0 & 1 & 0 & $R_1$\\
      0.36&0& 0 & 0 & 0 & 0 & $R_2$\\
      0 & 0 & 0 & 0 & 1 & 0 & $R_3$\\
      0 & 0 & 0 & 0 & 0 & 1 & $R_4$\\
      0 &9.4& 0 & 0 & 0 & 0 & $R_5$\\
      0 & 0 & 0 & 1 & 0 & 0 & $R_6$\\
      0 & 0 & 1 & 0 & 0 & 0 & $R_7$\\
      1 & 0 & 0 & 0 & 0 & 0 & $R_8$\\ 
    \end{block}
  \end{blockarray}
  \label{graph300} 
\end{equation}
\noindent On the other hand, each corresponding subsystem is weakly reversible and PL-RDK. Let $\widetilde{\mathscr N}_1$ be the network of kinetic complexes given by $\{\widetilde{R}_1: M_5\rightarrow 0.36M_1, \widetilde{R}_2: 0.36M_1 \rightarrow M_5, \widetilde{R}_3: R_3, \widetilde{R}_4:M_6\rightarrow 0.36M_1\}$ and $\widetilde{\mathscr N}_2$ be given by $\{\widetilde{R}_5: 9.4M_2\rightarrow M_1, \widetilde{R}_6: M_4 \rightarrow 9.4M_2, \widetilde{R}_7: R_7, \widetilde{R}_8: R_8\}$ (see \ref{good}).
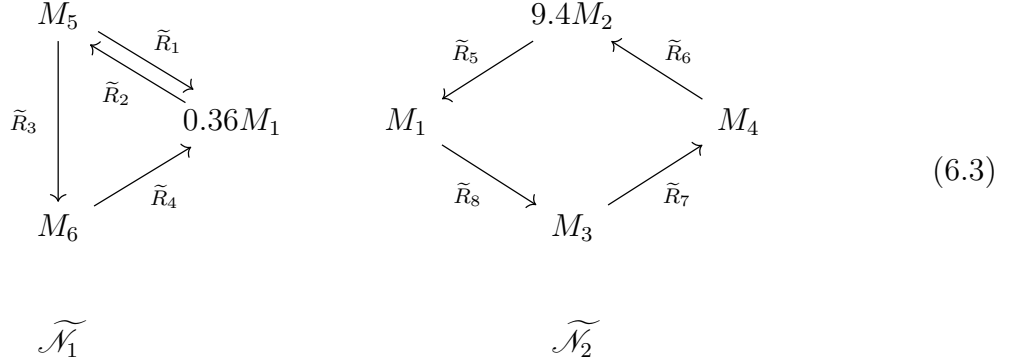
\begin{equation}\label{good}
    \begin{tikzcd}
M_5 \arrow[dd, "\widetilde{R}_3\;\;"'] \arrow[rd, "\widetilde{R}_1", shift left] &                                                   &                                    & 9.4M_2 \arrow[ld, "\widetilde{R}_5"'] &                                    \\
                                                                                 & 0.36M_1 \arrow[lu, "\widetilde{R}_2", shift left] & M_1 \arrow[rd, "\widetilde{R}_8"'] &                                       & M_4 \arrow[lu, "\widetilde{R}_6"'] \\
M_6 \arrow[ru, "\widetilde{R}_4"']                                               &                                                   &                                    & M_3 \arrow[ru, "\widetilde{R}_7"']    &                                    \\
\widetilde{\mathscr N}_1                                                         &                                                   &                                    & \widetilde{\mathscr N}_2              &                                   
\end{tikzcd}
\end{equation}
\noindent Accordingly, $\widetilde{\mathscr N}_{\mathscr D}=\widetilde{\mathscr N}_1\cup \widetilde{\mathscr N}_2$. Observe that $\widetilde{n}_{\mathscr D}=7$ and $\widetilde{\ell}_{\mathscr D}=2$ implying that $\widetilde{n}_{\mathscr D}-\widetilde{\ell}_{\mathscr D}=5$. Moreover, notice that $\widetilde{s}_{\mathscr D}=5, \widetilde{s}_1=2$, and $\widetilde{s}_2=3$ suggesting that, by Proposition \ref{tela}, the decomposition $\widetilde{\mathscr N}_{\mathscr D}=\widetilde{\mathscr N}_1\cup \widetilde{\mathscr N}_2$ is independent. 
\end{example}

\subsection{Statement and proof of the sufficient condition}\label{main_result}

We can now formulate and derive our main result:

\begin{theorem} \label{theorem13}
Let $(\mathscr N, K)$ be a weakly power law system with a complex balanced \textnormal{PL-RDK} decomposition $\mathscr D: \mathscr N=\mathscr N_1\cup \cdots \cup \mathscr N_k$ with $P_{Z,i}=\widetilde{S}_i$. If $\mathscr D$ is incidence independent and the induced covering $\widetilde{\mathscr D}$ is independent, then $(\mathscr N, K)$ is a weakly reversible \textnormal{CLP} system with $P_Z=\sum \widetilde{S}_i$.  
\end{theorem}

\begin{proof}
Each subnetwork of  $(\mathscr N, K)$ in the decomposition is weakly reversible and hence, the network is weakly reversible. For the last two statements, we prove the case when $k=2$ while the general case can be proven inductively.

Since each decomposition subnetwork is a complex balanced PL-RDK system, by Theorem \ref{muller-regensburger}, it is of CLP type, i.e., $Z_+(\mathscr N_i, K_i)\neq \varnothing$ and for $x_i \in Z_+(\mathscr N_i, K_i)$,
\begin{equation*}
    Z_+(\mathscr N_i, K_i)=\{x\in \mathbb R^{\mathscr S}_>| \log x-\log x_i \in \widetilde{S}_i^\perp \}.
\end{equation*}
\noindent Since the decomposition is incidence independent, we have 
\begin{equation*}
  Z_+(\mathscr N, K)=Z_+(\mathscr N_1, K_1)\cap Z_+(\mathscr N_2, K_2)  
\end{equation*}
\noindent Hence, $x^*\in Z_+(\mathscr N, K)$ if and only if
\begin{equation*}
    \log x^*\in \Big (\log x_1 + \widetilde{S}_1^{\perp}\Big )\cap \Big (\log x_2 + \widetilde{S}_2^{\perp}\Big ). 
\end{equation*}

\noindent From properties of cosets, 
\begin{equation*}
    \Big (\log x_1 + \widetilde{S}_1^{\perp}\Big )\cap \Big (\log x_2 + \widetilde{S}_2^{\perp} \Big)\neq \varnothing \Leftrightarrow \log x_1 - \log x_2 \in \widetilde{S}_1^{\perp}+\widetilde{S}_2^{\perp}. 
\end{equation*}
\noindent The independence of the induced decomposition ensures that $\widetilde{S}_1\cap \widetilde{S}_2=\{0\}$. Then, $\widetilde{S}_1^{\perp}+\widetilde{S}_2^{\perp}=\Big(\widetilde{S}_1\cap \widetilde{S}_2 \Big)^{\perp}=\{0\}^{\perp}=\mathbb R^{\mathscr S}$. Thus, 
\begin{equation*}
    \Big (\log x_1 + \widetilde{S}_1^{\perp}\Big )\cap \Big (\log x_2 + \widetilde{S}_2^{\perp} \Big)\neq \varnothing.
\end{equation*}
\noindent Let $\widehat {x}\in \Big (\log x_1 + \widetilde{S}_1^{\perp}\Big )\cap \Big (\log x_2 + \widetilde{S}_2^{\perp} \Big)$ and take $x^*=e^{\widehat {x}}$. We have $\widehat {x}\in \Big (\log x_1 + \widetilde{S}_1^{\perp}\Big )$ and  $\widehat {x}\in \Big (\log x_2 + \widetilde{S}_2^{\perp}\Big )$, hence $x^*\in  Z_+(\mathscr N_1, K_1) \cap  Z_+(\mathscr N_2, K_2)= Z_+(\mathscr N, K)$. To show the log parametrization, note that
\begin{equation*}
    \Big (\log x_1 + \widetilde{S}_1^{\perp}\Big )\cap \Big (\log x_2 + \widetilde{S}_2^{\perp} \Big)= \log x^* + \Big (\widetilde{S}_1^{\perp}+\widetilde{S}_2^{\perp}\Big ) = \log x^* + \Big (\widetilde{S}_1^{\perp}+\widetilde{S}_2^{\perp}\Big )^{\perp}. 
\end{equation*}
\noindent Hence, $Z_+(\mathscr N, K)= \bigg \{x\in \mathbb R^{\mathscr S}_{>0}| \log x - \log x^* \in \Big (\widetilde{S}_1^{\perp}+\widetilde{S}_2^{\perp}\Big )^{\perp}\bigg \}$.
\end{proof}

\begin{corollary}\label{corollary2}
    Let $(\mathscr N, K)$ be a \textnormal{PL-NDK} system satisfy the conditions of Theorem \ref{theorem13} for a decomposition into \textnormal{PL-NDK} subnetworks with zero kinetic deficiency, i.e., $\widetilde{\delta}=0$. Then the induced subnetwork has zero kinetic complex deficiency and $(\mathscr N, K)$ is unconditionally complex balanced, i.e., it is complex balanced for any set of rate constants.
\end{corollary}

\begin{proof}
Since the decomposition $\mathscr D$ is incidence independent, we have $n-\ell=\sum (n_i-\ell_i)$, and due to the surjective map $\mathscr C\rightarrow \widetilde{\mathscr C}$, $n_i-\ell_i\geq \widetilde{n}_{i}-\widetilde{\ell}_{i}$, we obtain the independence of the induced decomposition that the kinetic complex deficiency is less than the sum of the subnetworks' kinetic deficiencies which is zero. The unconditional complex balancing results from the conclusion of the theorem.
\end{proof}

\begin{example} 
Corollary \ref{corollary2} holds for a PL-RDK system satisfying the conditions of the theorem for a weakly reversible PL-TIK decomposition, since Talabis et al. have shown that such system have zero kinetic deficiency \cite{tala2019}. One could in fact extend the concept of a PL-TIK system (i.e., a system with zero kinetic reactant deficiency) to such a PL-NDK systems with the additional requirement that the induced decomposition be $\widetilde{R}$-independent, i.e., the augmented reactant subspaces of the subnetworks form a direct sum.
\end{example}

We remark that Corollary \ref{corollary2} suggests that further results on GMAS such as those of Cracium et al. on the generalized Birch's Theorem \cite{crac2019}, and Boros et al. \cite{boros2020} on linear stability  can also be extended to such PL-NDK systems.

\subsection{Discussion}

The generalization of the main result of Fortun et al. \cite{fort2019} is the following theorem:

\begin{theorem}\label{gen_dzt}
Let $(\mathscr N, K)$ be a power law systems with a weakly reversible \textnormal{PL-RDK} decomposition $\mathscr D:\mathscr N=\mathscr N_1\cup \cdots \cup \mathscr N_k$. If $\mathscr D$ is bi-level independent $($both the network decomposition and the induced subnetwork decomposition are independent$)$ and of \textnormal{PLP} type with $P_{E,i}=\widetilde{S}_i$, then $(\mathscr N, K)$ is a weakly reversible \textnormal{PLP} system with $P_E=\sum \widetilde{S}_i$.
\end{theorem}

A comparison with the main result immediately reveals that the two results are fully analogous. In fact, their proofs differ only in the use of Theorem \ref{feinberg decomp} in the PLP case and Theorem \ref{farinas decomp} for CLP systems. On the other hand, there are only a few known examples of PLP PL-RDK systems in contrast to the very general results of M\"{u}ller and Regensburger about CLP PL-RDK systems, so that currently our main result can be considered more impactful.

A further interesting point of discussion concerns research questions raised by our main result. The fact that the induced decomposition plays an essential role suggests that it may be interesting to study the set of common complexes of the induced decomposition and explore relationships between it and the common complexes of the original decomposition. These relationships could also be decomposition class specific, e.g., for PMM decompositions. Also fully unexplored are relationships between decomposition classes and conditions for independence of the induced decomposition.

Another set of research questions emerge from the observation that, so far, the results on the complex balancing of the whole networks are ``binar'' in the sense that Farinas et al. deal with $\mathscr C_{\mathscr D}=\varnothing$ and our main result with $\mathscr C_{\mathscr D}\neq\varnothing$  for any incidence independent decomposition. If one refined the latter set by considering only a particular decomposition class, again say PMM that are incidence independent, it might be possible to identify further classes of kinetics which enable complex balancing of the whole system.

\section{Summary and outlook}

We provide here a summary of this paper as well as some recommendations for further studies.

\begin{enumerate}
    \item  We investigated the incidence independence of decompositions using the set of common complexes of the subnetworks. A framework was created to characterize decomposition classes by their incidence independence properties.
    
    \item Interesting decomposition classes, identified by the sets of their subnetworks' common complexes, were introduced. Subclasses of these decomposition classes, which are incidence independent, were specified. 
    
    \item For future studies, other structural conditions sufficient or necessary for the incidence independence of other subclasses of PBS decompositions can be explored. Other decomposition classes identified by their sets of common complexes can also be considered and investigated.
    
    \item We identified a sufficient condition that guarantees the existence of complex balanced equilibria of some PLK systems with incidence independent and complex balanced decompositions.
    
    \item Through the result identified in 4, we obtained a generalization of the Deficiency zero Theorem for some PLK systems. 
    
    \item It is recommended to explore the capacity of other PLK systems to admit complex balanced equilibria that have incidence independent decompositions.
\end{enumerate}

\vspace{0.5cm}
\noindent \textbf{Acknowledgement}
L. L. Fontanil extends his gratitude to the Department of Science and Technology-Science Education Institute (DOST-SEI), Philippines for supporting him through the Accelerated Science and Technology Human Resource Development Program (ASTHRDP) scholarship grant. This would also not be possible without the significant comments and suggestions of Dr. Noel T. Fortun to the early forms of this paper.

\baselineskip=0.25in

\pagebreak

\appendix

\section*{Appendix}

The following table list the notations used in this paper that were adapted for subnetworks.
\begin{table}[H]
\centering
	\begin{tabular}{|c|c|c|}
		\hline
		\textbf{Notion} & \textbf{Parent Network}: $\mathscr N$  & \textbf{Subnetwork}: $\mathscr N_i$  \\ \hline
		Species, complex, & $\mathscr S, \mathscr C,$ and $\mathscr R$ & $\mathscr S_i, \mathscr C_i,$ and $\mathscr R_i$\\ 
		and reaction sets&&\\ \hline
		Number of species, complexes, & $m, n,$ and $p$ & $m_i,n_i,$ and $p_i$\\ 
		and reactions&&\\ \hline
		Stoichiometric subspace & $S$ & $S_i$ \\ \hline
		Network rank  & $s$ & $s_i$ \\ \hline
		Deficiency & $\delta$ & $\delta_i$ \\ \hline
		Number of linkage classes & $\ell$ & $\ell_i$ \\ \hline
		Incidence map/matrix & $I_a$ & $I_{a,i}$ \\ \hline

		Reaction network of & $\widetilde{\mathscr N}$ & $\widetilde{\mathscr N}_i$\\
		kinetic complexes& &\\\hline
		
		Set of kinetic complexes & $\widetilde{\mathscr C}$ & $\widetilde{\mathscr C}_i$ \\\hline		
		
		Kinetic order subspace & $\widetilde{S}$ & $\widetilde{S}_i$ \\\hline
		
    	Kinetic rank & $\widetilde{s}$ & $\widetilde{s}_i$ \\\hline
    	
    	Kinetic deficiency & $\widetilde{\delta}$ & $\widetilde{\delta}_i$ \\\hline
		
	\end{tabular}
	\caption*{Notations used for the parent network and the subnetworks.}
\end{table}
\end{document}